\documentclass[11pt]{amsart}
\usepackage{amssymb,mathrsfs,graphicx,enumerate}
\usepackage{amsmath,amsfonts,amssymb,amscd,amsthm,bbm}
\usepackage[retainorgcmds]{IEEEtrantools}
\usepackage{colortbl}
\usepackage{lipsum}
\usepackage{graphicx, subfigure}
\usepackage[caption=false]{subfig}
\usepackage{graphicx}
\title[Dynamical systems approach for the shape matching]{A dynamical systems approach for the shape matching of polytopes along rigid-body motions}

\author[Ha]{Seung-Yeal Ha}
\address[Seung-Yeal Ha]{\newline Department of Mathematical Sciences and Research Institute of Mathematics \newline Seoul National University, Seoul 08826, and \newline
Korea Institute for Advanced Study, Hoegiro 85, 02455, Seoul, Republic of Korea}
\email{syha@snu.ac.kr}

\author[Park]{Hansol Park}
\address[Hansol Park]{\newline Department of Mathematical Sciences\newline Seoul National University, Seoul 08826, Republic of Korea}
\email{hansol960612@snu.ac.kr}

\newtheorem{theorem}{Theorem}[section]
\newtheorem{lemma}{Lemma}[section]

\newtheorem{proposition}{Proposition}[section]
\newtheorem{remark}{Remark}[section]

\newtheorem{definition}{Definition}[section]

\newcommand{\bbr}{\mathbb R}

\begin{document}

\date{\today}


\thanks{\textbf{Acknowledgment.} The work of S.-Y. Ha was supported by National Research Foundation of Korea(NRF-2020R1A2C3A01003881)}

\begin{abstract}
We present a dynamical systems approach for geometric matchings in an ensemble of polytopes along rigid-body motions. Each polytope can be characterized by a vertex set and edge or faces determined by vertices, and polygons and simplexes correspond to a polytope. For a geometric matching, we propose a system of dynamical system for the evolution of centroids and rotations of polytopes to match the vertices under rigid-body motions which can be decomposed as a composition of translation and rotations. Our proposed dynamical system acts on the product space $({\mathbb R}^d \times SO(d))^N$. The evolution of centroids can be described by the coupled linear second-order dynamical system with diffusive linear couplings, whereas rotations for the matching of vertices are described by the Lohe matrix model on $SO(d)^N$. In particular, the Lohe matrix model has been derived from some set of physical principles compared to previous works in which the Lohe matrix model were employed as a  system dynamics. This is a contrasted difference between earlier works on the Lohe matrix model which has been adopted a priori for an aggregate modeling of matrices. We also provide an analytical result leading to the complete shape matchings for an ensemble of congruent polytopes, and several numerical examples to illustrate analytical results visually. 
\end{abstract}

\maketitle \centerline{\date}


\section{Introduction} \label{sec:1}
\setcounter{equation}{0} 
Collective behaviors of complex systems often appear in nature, e.g., aggregation of bacteria \cite{ F-P-P}, synchronization of fireflies and pacemaker cells \cite{B-C-M, B-B, Pe}, flocking of birds, swarming of fish, etc. For a survey, we refer to \cite{A-B, P-R,St, Wi1}. In this paper, we are interested in the rigid-body motions of polytopes leading to shape matchings. A polytope is a geometric object consisting of vertices, lines and faces connecting them, and we will identify each polytope with the set of vertices with fixed pairwise distance between them. Consider an ensemble consisting of vertex sets moving with rigid-body motions in free Euclidean space without any obstacles. Under this circumstance, we are mainly interested in the design of a dynamical system for shape matching among point sets via rigid-body motions. In most literature on the collective behaviors, mathematical modelings were mostly done for point particles without any internal structures, e.g., the particle Keller-Segel model \cite{F-P-P, F-Z}, the Winfree model \cite{Wi2}, the Kuramoto model \cite{C-H-J-K, C-S, D-X, D-B, D-B1, H-K-R1, H-L-X,Ku1, Ku2,V-M2, V-M1}, the Lohe sphere model \cite{C-C-H, C-H5, H-K-R0,O}, the Lohe matrix model \cite{D, H-K, H-K-R2, H-R, Lo-6, Lo-1, Lo-2} and the Lohe tensor model \cite{H-P, H-P-1} etc. 

To fix the idea, we consider an ensemble of congruent $n$-polytopes with different initial positions moving with a rigid-body motion in ${\mathbb R}^d$. Let $\Lambda$ be a fixed set of vertices of congruent $n$-polytopes and we denote by $\{x_\alpha^i \}_{\alpha \in \Lambda}$ be the vertex set of the $i$-th polytope. In this situation, we would like to address the following dynamic shape matching problems: \newline
\begin{itemize}
\item
(Q1) (Design of a dynamical system):~We design a continuous-time dynamical system leading to exact shape matching of all polytopes ${\mathcal C} := \{ \{x_\alpha^i \}_{\alpha \in \Lambda}:~i = 1, \cdots, N \}$ asymptotically: 
\[ \lim_{t \to \infty} \max_{\alpha \in \Lambda} \|x^i_\alpha - x^j_\alpha\| = 0, \quad \forall~i, j = 1, \cdots, N, \]
where $\| \cdot \|$ be a $\ell^2$-norm in $\bbr^d$. 

\vspace{0.1cm}

\item
(Q2) (Validity question):~If such dynamical system exists, what are the sufficient conditions for initial configurations and system parameters leading to the exact shape matching?
\end{itemize}

\vspace{0.1cm}

For each polytope, we decompose a rigid-body motion of each polytope as a direct sum of translation motion (external motion) and rotation (internal motion) so that 
\[ \{ \mbox{ Rigid-body motions consisting of translation and rotation} \}  \cong \bbr^d \times SO(d). \]

In this paper, we study above two questions (Q1) - (Q2).  More precisely, our main result can be summarized as follows. Our first main result deals with the design of a dynamical system on the state space $(\bbr^d \times SO(d))^N$. First, we describe our governing dynamical system as follows. For the $i$-th $n$-polytope $P^i = \{ x_\alpha^i \}_{\alpha \in \Lambda}$, we define its center-of-mass and relative displacements of vertices from the center-of-mass:
\[  {\bar x}^i := \frac{ 1}{|\Lambda|} \sum_{\alpha\in\Lambda} x^i_\alpha, \qquad r^i_\alpha := x^i_\alpha - {\bar x}^i,  \quad i = 1, \cdots, N,  \]
so that 
\begin{equation} \label{A-0-0-0}
  x^i_\alpha(t) =  {\bar x}^i(t) + r^i_\alpha(t),  \quad \alpha \in \Lambda, \quad  i = 1, \cdots, N.
\end{equation} 
Our governing system consists of two subsystems for the center-of-mass and its displacements. For the motion of center-of-mass $\{ {\bar x}^i \}$, we derive a second-order linear consensus model with a constant damping and diffusive linear consensus (see Section \ref{sec:3.1}):
\begin{equation} \label{A-0-0}
 m\ddot{{\bar x}}^i = -\gamma\dot{{\bar x}}^i + \frac{\kappa}{N}\sum_{k=1}^N({\bar x}^k- {\bar x}^i),
\end{equation}
where $\gamma$ and $\kappa$ are positive friction coefficient and coupling strength, respectively. 

In contrast, the dynamics of displacement $r_\alpha^i$ from ${\bar x}^i$ will be determined by that of matrix $O^i \in SO(d)$:
\begin{equation} \label{A-0-1}
 r_\alpha^i(t) := O^i(t) r_\alpha, \quad \alpha \in \Lambda, \quad i = 1, \cdots, N,
\end{equation}
where $\{r_\alpha\}_{\alpha\in\Lambda}\subset\bbr^d$ with $\sum_{\alpha\in\Lambda}r_\alpha=0$.
In Section \ref{sec:3.2}, we will derive the Lohe matrix model on $SO(d)^N$ without imposing it based on several set of physical principles:
\begin{equation} \label{A-0-2}
m\Big(\ddot{O}^i(O^i)^T+\dot{O}^i(\dot{O}^i)^T\Big) = -\gamma\dot{O}^i(O^i)^T + \frac{\kappa}{2N}\sum_{k=1}^N\big(O^k(O^i)^T-O^i(O^k)^T\big).
\end{equation}
Finally, we combine \eqref{A-0-0}, \eqref{A-0-1} and \eqref{A-0-2} to propose a dynamical system for $\{ ({\bar x}_\alpha^i, O^i) \}$:
\begin{align}\label{A-1}
\begin{cases}
\displaystyle m\ddot{{\bar x}}^i = -\gamma\dot{{\bar x}}^i + \frac{\kappa}{N}\sum_{k=1}^N({\bar x}^k- {\bar x}^i), \quad t > 0,~~i = 1, \cdots, N, \\
\displaystyle m\Big(\ddot{O}^i(O^i)^T + \dot{O}^i(\dot{O}^i)^T\Big) = -\gamma\dot{O}^i(O^i)^T + \frac{\kappa}{2N}\sum_{k=1}^N\big(O^k(O^i)^T-O^i(O^k)^T\big).
\end{cases}
\end{align}
Once we solve the above decoupled system \eqref{A-1}, we can recover the original state $x^i_\alpha$ by the relations \eqref{A-0-0-0} and \eqref{A-0-1}:
\[    x^i_\alpha(t) =  {\bar x}^i(t) + O^i(t) r_\alpha, \quad t > 0, \quad \alpha \in \Lambda,~~i = 1, \cdots, N, \]
where $r_\alpha$ is the reference position determined by the initial position. \newline

Note that for the case with $m = 0$ and $\gamma = 1$, system \eqref{A-1} formally reduces to 
\begin{equation}\label{A-2}
\begin{cases}
\displaystyle \dot{{\bar x}}^i  = \frac{\kappa}{N}\sum_{k=1}^N({\bar x}^k- {\bar x}^i), \quad r^i_\alpha(t)=O^i(t) r_\alpha, \\
\displaystyle \dot{O}^i(O^i)^T = \frac{\kappa}{2N}\sum_{k=1}^N\big(O^k(O^i)^T-O^i(O^k)^T\big).
\end{cases}
\end{equation}
Note that the second equation $\eqref{A-2}_2$ corresponds to the Lohe matrix model with zero natural frequency matrices. Hence,  surprisingly the Lohe matrix model appears naturally in the dynamics on $SO(d)^N$. In previous studies \cite{D-F-M1, D-F-M-T, H-K-L-N} on the orientation synchronization of multi-agent systems, several suitable matrix-valued aggregation models were used without any particular justification, whereas we derive the Lohe matrix model on $SO(d)^N$ starting from some physical arguments on rotations. This is a contrasted difference between earlier works on the matrix-valued consensus model and the result obtained in the current paper. 

\vspace{0.5cm}

The rest of this paper is organized as follows. In Section \ref{sec:2}, we briefly review emergent dynamics of the second-order and first-order consensus models on $(\bbr^d)^N$ and $SO(d)^N$ with linear diffusive couplings and linear damping force. In Section \ref{sec:3}, we provide a derivation of dynamical systems for the displacement around center-of-mass beginning from Newton's second law. In Section \ref{sec:4}, we briefly discuss emergence of complete shape matching for systems \eqref{A-1} and \eqref{A-2} as a direct corollary of earlier results on the linear consensus model and the Lohe matrix model for identical polytopes. For similar and heterogeneous collections of polytopes, we also provide variant systems for shape matchings. In Section \ref{sec:5}, we provide several numerical simulations and compare them with our analytical results in Section \ref{sec:4}. Finally Section \ref{sec:6} is devoted to a brief summary of our main results and some discussion on a future work. 

\section{Preliminaries} \label{sec:2}
\setcounter{equation}{0}
In this section, we briefly discuss some first-order and second-order consensus models on three types of finite-dimensional manifolds ${\mathcal M}$:
\[ {\mathcal M}~:~\bbr^d, \quad SO(d), \quad {\mathcal M}_1 \times {\mathcal M}_2, \]
where ${\mathcal M}_i,~i= 1,2$ is a Riemannian manifold.

\subsection{A second-order linear consensus model} \label{sec:2.1} Let $q_i = q_i(t) \in \bbr^d$ be a quantifiable state of the $i$-th agent that we look for consensus. We assume that the dynamics of $q_i$ is governed by the Cauchy problem of the second-order linear consensus model:
\begin{equation}
\begin{cases} \label{B-1}
\displaystyle m {\ddot q}^i = -\gamma \dot{q}^i +\frac{\kappa}{N}\sum_{k=1}^N(q^k-q^i), \quad t > 0, \\
\displaystyle q^i(0) = q^{i0},  \quad i = 1, \cdots, N,
\end{cases}
\end{equation}
where $m$, $\gamma$, and $\kappa$ are nonnegative constants representing strengths of inertia, friction and coupling strength, respectively. Since system \eqref{B-1} is a linear system with constant coefficients, the explicit solution to \eqref{B-1} can be found. For the consensus dynamics of \eqref{B-1}, we introduce a configuration diameter ${\mathcal D}({\mathcal Q})$:
\begin{equation}  \label{B-2}
{\mathcal Q} := \{q^i \}_{i=1}^{N}, \quad {\mathcal D}({\mathcal Q}) := \max_{1 \leq i,j \leq N} \| q^i - q^j \|. 
\end{equation}
Note that 
\[ \mbox{asymptotic complete consensus occurs} \quad \Longleftrightarrow \quad \lim_{t \to \infty} {\mathcal D}({\mathcal Q}(t)) = 0. \]
Then, emergent dynamics of \eqref{B-1} in the following proposition.
\begin{proposition} \label{P2.1} 
Let ${\mathcal Q} = \{ q^i \}$ be a solution to the Cauchy problem \eqref{B-1}. Then, the following assertions hold.
\begin{enumerate}
\item
(Positive inertia): Suppose system parameters satisfy
\[ m > 0, \quad \gamma > 0, \quad \kappa > 0. \]
Then, ${\mathcal D}({\mathcal Q})$ satisfies 
\[
\mathcal{D}({\mathcal Q}) \lesssim
\begin{cases}
\exp\left(\displaystyle\frac{-\gamma+\sqrt{\gamma^2-4m\kappa}}{2m}~ t\right), \quad &\gamma^2-4m\kappa\geq0, \\
\exp\left(-\displaystyle\frac{\gamma}{2m}~t\right), \quad &\gamma^2-4m\kappa<0.
\end{cases}
\]

\vspace{0.2cm}

\item
(Zero inertia): Suppose system parameters satisfy
\[ m = 0, \quad \gamma = 1, \quad \kappa > 0. \]
Then, ${\mathcal D}({\mathcal Q})$ satisfies 
\[ {\mathcal D}({\mathcal Q}) \lesssim \exp\left(-\frac{\kappa}{\gamma} t\right). \]
\end{enumerate}
\end{proposition}
\begin{proof}
\noindent (i)~We set the transversal displacement $x_{ij}$  as follows.
\[ q^{ij} :=q^i-q^j, \quad \forall~ i, j=1, 2, \cdots, N. \]
By the refining relation \eqref{B-2}, it suffices to check that $q^{ij}$ tends to zero exponentially fast. Then, the transversal difference $q^{ij}$ satisfies
\[
m\ddot{q}^{ij}+\gamma\dot{q}^{ij}+\kappa q^{ij}=0.
\]
By direct calculation, one has 
\[
q^{ij}(t)=c_1\exp\left(\frac{-\gamma+\sqrt{\gamma^2-4m\kappa}}{2m}~t\right)+c_2\exp\left(\frac{-\gamma-\sqrt{\gamma^2-4m\kappa}}{2m}~t\right),
\]
where $c_1$ and $c_2$ are constants determined by initial data. \newline

\noindent $\bullet$~Case A $( \gamma^2-4m\kappa\geq0)$:~In this case, since
\[
\frac{-\gamma-\sqrt{\gamma^2-4m\kappa}}{2m}\leq \frac{-\gamma+\sqrt{\gamma^2-4m\kappa}}{2m}<0,
\]
one has 
\[
|q^{ij}(t)| \lesssim \exp\left(\displaystyle\frac{-\gamma+\sqrt{\gamma^2-4m\kappa}}{2m}~t\right), \quad t \geq 0.
\]

\vspace{0.5cm}

\noindent $\bullet$~Case B $( \gamma^2-4m\kappa < 0)$:~In this case, one has 
\[
\mathrm{Re}\left(\frac{-\gamma-\sqrt{\gamma^2-4m\kappa}}{2m}\right)=\mathrm{Re}\left(\frac{-\gamma+\sqrt{\gamma^2-4m\kappa}}{2m}\right)=-\frac{\gamma}{2m}<0.
\]
Thus we have
\[
|q^{ij}(t) | \lesssim \exp\left(-\displaystyle\frac{\gamma}{2m}~t\right), \quad t \geq 0.
\]

\vspace{0.5cm}

\noindent (ii)~Consider the case:
\[ m = 0, \quad \gamma > 0. \]
In this case, $q^{ij}$ implies
\[
\dot{q}^{ij} =-\frac{\kappa}{\gamma} q^{ij}, \quad \mbox{i.e.,} \quad q^{ij}(t)=e^{-\frac{\kappa}{\gamma} t} q^{ij}(0).
\]
\end{proof}

\subsection{The second-order Lohe matrix model on $SO(d)$} \label{sec:2.2}
Let $A^i \in SO(d)$ be an orthogonal matrix whose dynamics is governed by the Cauchy problem to the second-order matrix model:
\begin{equation} \label{B-3}
\begin{cases}
\displaystyle m \Big[ \ddot{A}^i(A^{i})^{T} + \dot{A}^i(\dot{A}^i)^{T} \Big ] = -\gamma \dot{A}^i (A^i)^{T}  +  \frac{\kappa}{2N}\sum_{j=1}^{N}\left[A^j (A^i)^{T} - A^i (A^j)^{T}\right], \quad t > 0, \\
(A^i, {\dot A}^i)(0)= (A^{i0}, {\dot A}^{i0}), \quad i = 1, \cdots, N,
\end{cases}
\end{equation}
subject to constraints:
\begin{equation} \label{B-4}
A^{i0} \in SO(d), \quad \dot{A}^{i0} (A^{i0})^{T}+ A^{i0} (\dot{A}^{i0})^{T} =O_d, \quad i = 1, \cdots, N,
\end{equation}
where $O_d$ is the $d\times d$ zero matrix.  The matrix model \eqref{B-3} was first introduced in \cite{H-K}, and its emergent dynamics was also extensively studied there.

\begin{lemma} \label{L2.1}
Let ${\mathcal A} = \{ A^i \}$ be a solution to \eqref{B-3} - \eqref{B-4}. Then, one has 
\[  A^i(t) (A^i(t))^T = I_d, \quad t \geq 0,\quad i=1,\cdots,N, \]
where $I_d$ is the $d \times d$ identity matrix. 
\end{lemma}
\begin{proof} We first rewrite $\eqref{B-3}_1$ as 
\begin{equation} \label{B-5}
m\ddot A^i (A^i)^T  = - m {\dot A}^i ({\dot A}^i)^T  - \gamma {\dot A}^i (A^i)^T  +  \frac{\kappa}{2N} \sum_{k=1}^N \Big( A^k (A^i)^T - A^i (A^k)^T \Big).
\end{equation}
For the desired estimate, we set 
\begin{equation} \label{B-5-0}
M^i := A^i (A^i)^{T}, \quad i=1,\cdots,N.
\end{equation}
Then, by the assumptions \eqref{B-4} on initial data, we have
\begin{equation} \label{B-5-1}
M^i(0) = I_d, \quad {\dot M}^i(0)= O_d, \quad i = 1, \cdots, N. 
\end{equation}
For the desired estimate, we first claim:
\begin{equation*} \label{B-5-2}
{\dot M}^i(t) = 0, \quad t > 0, \quad i = 1, \cdots, N. 
\end{equation*}

\vspace{0.3cm}

\noindent $\bullet$ Step A (Derivation of a dynamics for $M_i$): We differentiate the relation \eqref{B-5-0} twice with respect to $t$ to obtain 
\begin{equation} \label{B-5-3}
\dot M^i  = {\dot A}^i (A^i)^T + A^i ({\dot A}^i)^T, \qquad {\ddot M}^i = {\ddot A}^i (A^i)^T + A^i ({\ddot A}^i)^T + 2 {\dot A}^i ({\dot A}^i)^T.
\end{equation}
We take a transpose of \eqref{B-5} to obtain
\begin{equation} \label{B-6}
m A^i ({\ddot A}^i)^T = - m {\dot A}^i ({\dot A}^i)^T - \gamma A^i ({\dot A}^i)^T + \frac{\kappa}{2N} \sum_{k=1}^N \Big (A^i (A^k)^T - A^k (A^i)^T \Big).
\end{equation}
We add \eqref{B-5} and \eqref{B-6} to see 
\begin{equation} \label{B-6-1}
m \Big( {\ddot A}^i  (A^i)^T + A^i ({\ddot A}^i)^T  \Big) = -2m {\dot A}^i ({\dot A}^i)^T  -\gamma \Big({\dot A}^i (A^i)^T + A^i ({\dot A}^i)^T \Big ).
\end{equation}
Then, we use \eqref{B-5-3} and \eqref{B-6-1} to see that $M^i$ satisfies
\begin{equation} \label{B-7}
m {\ddot M}^i = -\gamma {\dot M}^i.
\end{equation}

\vspace{0.5cm}

\noindent $\bullet$ Step B ($M_i$ is a constant of motion): We use \eqref{B-7} and $\| A \|_F^2 = \mbox{Tr}(A A^T)$ to see
\begin{align*}
\begin{aligned}
m\frac{d}{dt} \|\dot M^i\|_\text{F}^2 &= m\frac{d}{dt}\, \text{Tr}(\dot M^i ({\dot M}^i)^T) = m\,\text{Tr} \Big( \ddot M^i ({\dot M}^i)^T + {\dot M}^i ({\ddot M}^i)^T \Big) \\
&=\mbox{Tr}\Big(-\gamma {\dot M}^i ({\dot M}^i)^T -\gamma {\dot M}^i ({\dot M}^i)^T \Big) =-2\gamma \mbox{Tr}\Big( {\dot M}^i ({\dot M}^i)^T \Big) =-2\gamma \|\dot M^
i\|_\text{F}^2.
\end{aligned}
\end{align*}
This yields
\begin{equation} \label{B-8}
\|\dot M^i(t)\|_\text{F}^2 = \|\dot M^i(0)\|_\text{F}^2e^{-\frac{2\gamma}{m}t},\quad t>0.
\end{equation}
Next, we use \eqref{B-5-1} and \eqref{B-8} to get 
\begin{equation*} \label{B-8-1}
{\dot M}^i(t) = O_d, \quad t > 0.
\end{equation*}
Again, this and \eqref{B-5-1} imply the desired estimate:
\[
M^i(t)  = I_d, \quad t \geq 0.
\]
\end{proof}
\begin{proposition} \label{P2.2}
\emph{\cite{H-K, H-K-R2}}
Let ${\mathcal A} = \{ A^i \}$ be a solution to \eqref{B-3} - \eqref{B-4}. Then the following assertions hold.
\begin{enumerate}
\item
Suppose system parameters satisfy
\[ m > 0, \quad \gamma > 0, \quad \kappa > 0. \]
Then, ${\mathcal D}({\mathcal A})$ satisfies 
\[
\lim_{t\to\infty} {\mathcal D}({\mathcal A}(t))=0.
\]
\item
Suppose system parameters and initial data satisfy 
\[ m = 0, \quad \gamma > 0, \quad \kappa > 0, \quad  \max_{i,j} \|A^{i0} - A^{j0} \|_{F} < 1. \]
Then, ${\mathcal D}({\mathcal A})$ satisfies 
\[ {\mathcal D}({\mathcal A}(t)) \lesssim  e^{-\frac{\kappa}{\gamma} t}. \]
\end{enumerate}
\end{proposition}
\begin{proof}
The detailed proofs can be found in \cite{H-K} and \cite{H-K-R2}, respectively. However, for reader's convenience, we briefly sketch main ideas to get some feeling how it goes. \newline

\noindent (i) We define some energy functionals as follows:
\[
{\mathcal E}(t) =\mathcal{E}[{\mathcal A}(t), {\dot {\mathcal A}}(t)] :=\frac{m}{N}\sum_{i=1}^N\|\dot{A}^i\|^2_F+\frac{\kappa}{2N^2}\sum_{i, k=1}^N\|A^i-A^k\|_F^2.
\]
Then, ${\mathcal E}$ is nonnegative, and by direct calculation (Proposition 3.1 \cite{H-K}), one has 
\[
\frac{2\gamma}{N}\int_0^t \sum_{i=1}^N \|\dot U_i(s)\|_\text{F}^2 \,ds = \mathcal E(0) - \mathcal E(t) \leq \mathcal E(0), \quad t>0.
\]
This implies 
\[  \int_0^{\infty} \|\dot U_i(s)\|_\text{F}^2 \,ds < \infty. \]
Moreover, we can check that  $\Big| \frac{d}{dt} \| {\dot U}_i \|_{F} \Big|$ is uniformly bounded. Hence $ \| {\dot U}_i \|_{F}$ is Lipschitz continuous which clearly implies the uniform continuity of $\|\dot U_i\|_\text{F}$. Then, we can apply Barbalat's lemma to derive the desired zero convergence of ${\mathcal D}({\mathcal A})$ without any detailed decay rate. 

\vspace{0.2cm}

\noindent (ii) By detailed delicate estimate  (Lemma 3.1 \cite{H-K-R2}), the ensemble diameter ${\mathcal D}({\mathcal A})$ satisfies
\[ \frac{d}{dt} {\mathcal D}({\mathcal A})\leq -\kappa {\mathcal D}({\mathcal A}) (1 - {\mathcal D}({\mathcal A})), \quad \mbox{a.e.,}~t > 0. \]
Then, we integrate the above differential inequality to get 
\[ {\mathcal D}({\mathcal A}(t)) \leq \frac{{\mathcal D}({\mathcal A}^0)}{(1-{\mathcal D}({\mathcal A}^0))e^{\kappa t}+ {\mathcal D}({\mathcal A}^0)}, \quad t \geq 0.      \]
As long as ${\mathcal D}({\mathcal A}^0) < 1$, we have an exponential decay estimate of ${\mathcal D}({\mathcal A})$.

\end{proof}

\subsection{A consensus model on a product manifold} \label{sec:2.3}
Let $({\mathcal M}_1, g_1)$ and $({\mathcal M}_2, g_2)$ be two Riemannian manifolds respectively. Then the product $ {\mathcal M}_1\times {\mathcal M}_2$ defined by
\[ {\mathcal M}_1\times {\mathcal M}_2  = \{(x, y):~x\in {\mathcal M}_1,~~y\in {\mathcal M}_2\} \]
is a product Riemannian manifold with a product metric $g_1\oplus g_2$. 
\begin{definition}
\emph{\cite{F-P-P}}
 Let ${\mathcal X}(t)=\{x^i(t)\}_{i=1}^N\subset {\mathcal M}_1$ and ${\mathcal Y}(t)=\{y^i(t)\}_{i=1}^N$ be consensus flows on $({\mathcal M}_1, g_1)$ and $({\mathcal M}_2, g_2)$, respectively. Then the product flow ${\mathcal Z}(t)=\{z_i(t)=(x_i(t), y_i(t))\}_{i=1}^N\subset {\mathcal M}_1\times {\mathcal M}_2$ is also a consensus flow on ${\mathcal M}_1\times {\mathcal M}_2$. 
\end{definition}

\begin{proposition}
\emph{\cite{F-P-P}} 
Let $\{x_i\}_{i=1}^N$ and $\{y_i\}_{i=1}^N$ be consensus flows on $({\mathcal M}_1, g_1)$ and $({\mathcal M}_2, g_2)$ satisfying
\[
d_1(x_i(t), x_j(t))\to0\quad\mbox{and}\quad d_2(y_i(t), y_j(t))\to0\quad\mbox{as}\quad t\to\infty, \quad i, j \in \{1, \cdots, N \}. 
\]
Then, the product flow $\{z_i :=(x_i, y_i)\}_{i=1}^N$ on ${\mathcal M}_1\times {\mathcal M}_2$ satisfies 
\[
\lim_{t \to \infty} (d_1\oplus d_2)(z_i(t), z_j(t)) = 0, \quad \forall~ i, j  = 1, \cdots, N,
\]
i.e., asymptotic consensus emerges.
\end{proposition}
\begin{remark}
As aforementioned in Introduction, our governing system \eqref{A-1} corresponds to the consensus model on the product space ${\mathbb R}^{d N} \times SO(d)^N$. Therefore to derive a consensus estimate, it suffices to verify that each subsystem exhibits consensus estimate.
\end{remark}

\section{A consensus algorithm for congruent polytopes} \label{sec:3}
\setcounter{equation}{0}
In this section, we provide a heuristic derivation of the second subsystem of the consensus model \eqref{A-1} for rotations. Although $n$-polytope is completely characterized by the vertices, edges and faces, since we are considering the same shape of polytopes, we simply identify an $n$-polytope by their vertex set consisting of $n$ points in $\bbr^d$. \newline

Consider rigid-body motions of $N$ ensemble of congruent $n$-polytopes or vertex sets with size $n$ in $\bbr^d$. For this, we decompose a rigid body motion as a translation component and a rotation  component which can be characterized by a vector in $\bbr^d$ and a rotation matrix in $SO(d)$, respectively, i.e., set of all rigid-body motions consisting of translations and rotations is isomorphic to $\bbr^d \oplus SO(d)$. To fix the idea, let $\Lambda$ be a fixed index set with $|\Lambda| = n$ and let $X^i := \{x^i_\alpha \}_{\alpha \in \Lambda}$ be the vertex set of the $i$-th set. 

Note that if $X^i$ undergoes a rigid body motion, then the relative distances between vertices contained in the same polytope remain to be constant over time:
\[
\|x^i_\alpha(t)-x^i_\beta(t)\|= \|x^i_\alpha(0)-x^i_\beta(0)\|, \quad \forall~ t > 0, \quad \forall~\alpha, \beta\in\Lambda. 
\]
The motions of center-of-mass points and displacements from the center-of-mass will be taken care by translation part and rotation part, respectively. For a set $X^i = \{x^i_\alpha\}_{\alpha \in \Lambda}$, we set
\begin{equation} \label{C-0}
 {\bar x}^i(t)=\frac{1}{n}\sum_{\alpha\in\Lambda}x^i_\alpha(t), \qquad r_\alpha^i(t) := x_\alpha^i(t) - {\bar x}^i(t), \quad t \geq 0, \quad i = 1, \cdots, N, \quad \alpha \in \Lambda
 \end{equation} 
 Then, it follows from the rigid-body motion that for each $i = 1, \cdots, N$, there exits $O^i(t)\in\mathrm{SO}(d)$ with
\[
r^i_\alpha(t)=O^i(t) r_\alpha, \quad t \geq 0,
\]
where $r_\alpha$ is the reference configuration of congruent polytopes.  To sum up, one has 
\[  x_\alpha^i={\bar x}^i+O^ir_\alpha, \quad \forall~\alpha\in\Lambda,\quad i\in\{1, 2, \cdots, N\}. \]
In the sequel, we hueristically derive a coupled system for $({\bar x}^i, O^i)$:
\begin{align*}\label{C-1}
\begin{cases}
\displaystyle m\ddot{{\bar x}}^i+\gamma\dot{{\bar x}}^i-\displaystyle\frac{\kappa}{N}\sum_{k=1}^N({\bar x}^k- {\bar x}^i)=0, \quad r^i_\alpha(t)=O^i(t) r_\alpha, \\
\displaystyle m\Big(\ddot{O}^i(O^i)^T+\dot{O}^i(\dot{O}^i)^T\Big)+\gamma\dot{O}^i(O^i)^T-\displaystyle\frac{\kappa}{2N}\sum_{k=1}^N\big(O^k(O^i)^T-O^i(O^k)^T\big)=0,\\
\displaystyle ({\bar x}^i, \dot{\bar x}^i) \Big|_{t= 0+} = ({\bar x}^{i0}, \dot{\bar x}^{i0}) \quad \mbox{and} \quad 
({\bar O}^i, \dot{\bar O}^i) \Big|_{t= 0+} = ({\bar O}^{i0}, \dot{\bar O}^{i0}),  \\
 O^{i0}\in\mathrm{SO}(d),\quad \dot{O}^{i0} (O^{i0})^T+O^{i0} (\dot{O}^{i0})^T= O_d,\quad \forall~i = 1, \cdots, N. 
\end{cases}
\end{align*}

\vspace{0.2cm}

Consider the dynamics of $x^i_\alpha$ via Newton's second law. For this, we begin with 
\begin{equation} \label{C-2}
 m\frac{d^2x^i_\alpha}{dt^2} = F^i_{\alpha, a} +  F^i_{\alpha,c}, 
\end{equation} 
where $F^i_{\alpha,a}$ and $F^i_{\alpha,c}$ are {\it dissipative} alignment force and {\it conservative} configuration matching force acting on the vertex $x^i_\alpha$ due to vertex-vertex interactions, respectively. In the sequel, we impose following constraints on $F^i_{\alpha,c}$: \begin{equation} \label{C-3}
 \sum_{\alpha \in \Lambda}  F^i_{\alpha,c} = 0, \quad \|x^i_\alpha(t) -x^i_\beta(t)\| =   \|x^i_\alpha(0) -x^i_\beta(0)\|,
\end{equation} 
for $i = 1, \cdots, N$ and $t \geq 0$. \newline

In the following two subsections, we discuss explicit forms for $ F^i_{\alpha, a}$ and $F^i_{\alpha,c}$. 
\subsection{Dissipative alignment force} \label{sec:3.1}
In this subsection, we study a derivation of alignment force for the center-of-masses. As a dissipative aggregation force, we take the following ansatz:
\begin{equation}\label{C-4}
F^i_{\alpha, a} := \underbrace{-\gamma\frac{dx^i_\alpha}{dt}}_{\mbox{frictional force}} + \underbrace{\frac{\kappa}{N} \sum_{k=1}^N \sum_{\beta\in\Lambda} c_{\alpha\beta}(x^k_\beta-x^i_\alpha)}_{\mbox{distributive alignment force}}.
\end{equation}
The reason for employing the frictional force in the right-hand side of \eqref{C-4} is to make sure that the motion of center-of-mass becomes stationary asymptotically, i.e., without the frictional force, we will have oscillatory motions like harmonic oscillators, whereas the distributive alignment force is employed for the aggregation of corresponding vertex points. Of course, one can use more sophisticated nonlinear alignment as in \cite{H-H-K} for finite-time or algebraically slow alignments. 

Recall that our purpose here is to derive a consensus model  for vertices moving with rigid-body motions. To motivate communication weight $c_{\alpha \beta}$ or network topology, we consider two congruent polytopes moving with rigid-body motions and try to make corresponding vertices coincide together, i.e., if $\alpha=\beta$, we want to make 
\[ \lim_{t \to \infty} \|x_\alpha^i(t)-x_\alpha^j(t) \| = 0. \] 
On the other hand, if $\alpha\neq\beta$ with $ |x_\alpha^i-x_\beta^j |\rightarrow0$, then $|x_\alpha^i-x_\alpha^j|$ can not converge to 0 asymptotically. This can be seen as follows. Suppose there exists $\alpha\neq\beta$ such that  
\[  \lim_{t \to \infty} \|x_\alpha^i(t)-x_\beta^j(t)\| = 0. \]
Then, by triangle inequality, 
\[  \| x_\alpha^i - x_\alpha^j \| \geq \Big| \|x_\alpha^i - x_\beta^j \| - \|x_\alpha^j - x_\beta^j\|   \Big|.   \]
Letting $t \to \infty$, one derive a contradiction:
\[ 0 \geq \|x_\alpha^j(0) - x_\beta^j(0)\| > 0.    \]
Thus, we want to introduce some repulsion between $x_\alpha^i$ and $x_\beta^j$ for $\alpha\neq\beta$. Based on these intuitive arguments, we assume
\begin{equation} \label{C-4-1}
c_{\alpha\beta} >0 \quad \mbox{for $\alpha=\beta$}; \quad c_{\alpha\beta}<0, \quad \mbox{for $\alpha\neq\beta$}. 
\end{equation}
To be consistent with \eqref{C-4-1}, we set 
\begin{equation*} \label{C-4-2}
c_{\alpha\beta} =\delta_{\alpha\beta}(f(1)-f(0))+f(0),
\end{equation*}
where $f(0), f(1)$ and $\delta_{\alpha \beta}$ satisfy the following relations:
\[  f(0)\leq 0\leq f(1) \quad \mbox{and} \quad \delta_{\alpha \beta} = \begin{cases}
1, \quad & \alpha = \beta, \\
0, \quad & \alpha \neq \beta.
\end{cases}
\] 
Then, it is easy to see that 
\[ c_{\alpha \alpha} = f(1) \geq 0, \quad c_{\alpha \beta} = f(0) \leq 0 \quad \mbox{for $\alpha \neq \beta$}. \]
On the other hand, for the derivation of dynamics for ${\bar x}^i$, we sum up \eqref{C-2} using \eqref{C-3} to see
\begin{equation} \label{C-5}
m\sum_{\alpha\in\Lambda}\frac{d^2x_\alpha^i}{dt^2} = \sum_{\alpha\in\Lambda}F_{\alpha,a}^i.
\end{equation}
Note that the configuration matching force disappears in the R.H.S. of \eqref{C-5}. Next, we use \eqref{C-0}, \eqref{C-4} and \eqref{C-5} to see 
\begin{align*}
\begin{aligned} \label{C-6}
m\frac{d^2{\bar x}^i}{dt^2}&=\frac{m}{n}\sum_{\alpha\in\Lambda}\frac{d^2 x_\alpha^i}{dt^2}=\frac{1}{n}\sum_{\alpha\in\Lambda}F_{\alpha, a}^i \\
&=\frac{1}{n}\sum_{\alpha\in\Lambda}\left(-\gamma\frac{dx^i_\alpha}{dt}+\frac{\kappa}{N}\sum_{\beta\in\Lambda}\sum_{k=1}^Nc_{\alpha\beta}(x^k_\beta-x^i_\alpha)\right)\\
&=-\gamma\frac{d{\bar x}^i}{dt}+\frac{\kappa}{nN}\sum_{\alpha,\beta\in\Lambda}\sum_{k=1}^N \Big(\delta_{\alpha\beta}(f(1)-f(0))+f(0) \Big)(x_\beta^k-x_\alpha^i)\\
&=-\gamma\frac{d{\bar x}^i}{dt}+\frac{\kappa}{nN}\sum_{k=1}^N\left((f(1)-f(0))\sum_{\alpha\in\Lambda}(x_\alpha^k-x_\alpha^i)+f(0)\sum_{\alpha, \beta\in\Lambda}(x_\beta^k-x_\alpha^i)\right)\\
&=-\gamma\frac{d{\bar x}^i}{dt}+\frac{\kappa}{N}\sum_{k=1}^N\left((f(1)-f(0))( {\bar x}^k-{\bar x}^i)+nf(0)( {\bar x}^k- {\bar x}^i)\right)\\
&=-\gamma\frac{d{\bar x}^i}{dt}+(f(1)+(n-1)f(0))\cdot\frac{\kappa}{N}\sum_{k=1}^N({\bar x}^k-{\bar x}^i).
\end{aligned}
\end{align*}
For the alignment of center-of-masses, we set 
\[ \eta=f(1)+(n-1)f(0) > 0, \quad {\tilde \kappa} := \eta \kappa \]
and derive a second-order linear consensus model for center-of-masses:
\begin{equation}\label{C-7}
m\ddot{\bar x}^i =-\gamma \dot{\bar x}^i+\frac{{\tilde \kappa}}{N}\sum_{k=1}^N ({\bar x}^k- {\bar x}^i).
\end{equation}
Note that for a sufficiently small $m$, the second-order system \eqref{C-7} can be well-approximated by the corresponding first-order linear consensus model by the direct application of  Tikhonov's theory (see \cite{H-S} for a related problem):
\begin{align*}
\dot{{\bar x}}^i=\frac{\tilde{\kappa}}{\gamma N}\sum_{k=1}^N ({\bar x}^k- {\bar x}^i).
\end{align*}
\subsection{Conservative configuration force} \label{sec:3.2} In this subsection, we discuss the derivation of dynamics for  the displacement $r_\alpha^i$, in other words, the dynamics of $O^i_\alpha$. Since initial configuration point set $X^i(0)$ are congruent for all $i$ and they tend to the same position asymptotically after translations and rotations, there exist constant vectors $r_\alpha$ with $\alpha\in\Lambda$ such that 
\[
r_\alpha^i(t)={O}^{i}(t) r_\alpha, \quad O^i \in SO(d),
\]
for all $i=1, 2, \cdots, N$ and $\alpha\in\Lambda$. Hence, to see the motion of displacement $r_\alpha^i$, we only need to know the governing system for the orthogonal matrix $O^i$. Once we know $O^i$, we can determine the position of $x_\alpha^i$:
\begin{align}\label{dec}
x_\alpha^i(t)={\bar x}^i(t)+O^i(t)r_\alpha.
\end{align}

\vspace{0.2cm}

In the following lemma, we consider Netwon's system with a conservative forcing $F$:
\begin{equation} \label{C-3-1}
 {\dot x} = v, \quad m{\dot v} = F, \quad t > 0. 
\end{equation} 
\begin{lemma} \label{L3.1}
Suppose that the conservative force $F$ takes the following form:
\begin{equation} \label{C-3-1-1}
  F = A x \quad  A \in {\mathbb R}^{d \times d}, 
\end{equation} 
and let $(x, v)$ be a solution to \eqref{C-3-1}. Then, the matrix $A$ is symmetric:
\[ A^T  = A. \]
\end{lemma}
\begin{proof} Let $E_k$ and $E_p$ be the kinetic and potential energies, respectively. Note that the kinetic energy $E_k$ takes a definite form:
\begin{equation} \label{C-3-2}
 E_k = \frac{1}{2} m \langle v, v \rangle, 
\end{equation} 
where $\langle \cdot, \cdot \rangle$ is a standard inner product in ${\mathbb R}^d$. Then, it follows from \eqref{C-3-1} and \eqref{C-3-2} that  
\begin{equation} \label{C-3-3}
\frac{d}{dt}E_k=v\cdot F.
\end{equation}
Let $E_p$ be a potential energy related to system \eqref{C-3-1}, and we set the total energy $E$ as 
\[ E := E_k + E_p. \]
Then, by energy conservation law, one has 
\begin{equation} \label{C-3-4}
\frac{d}{dt}(E_k+E_p)=0.
\end{equation}
Now we use the special ansatz for $F$ and decompose the matrix $A$ as a sum of symmetric and skew-symmetric parts:
\begin{equation} \label{C-3-5}
F =B x +Cx,
\end{equation}
where
\[
B =\mathrm{Sym}(A) = \frac{1}{2} (A + A^T),\qquad C= \mathrm{Skew}(A) = \frac{1}{2} (A - A^T).
\]
Note that the relations \eqref{C-3-3}, \eqref{C-3-4} and \eqref{C-3-5} imply
\[
\frac{dE_k}{dt} =\langle \dot{x}, Bx \rangle+\langle \dot{x}, C x\rangle.
\]
Since $B$ is a symmetric matrix, the above relation can be written as 
\[
\frac{d}{dt}E_k=\frac{d}{dt}\left(\frac{1}{2}\langle x, B x \rangle\right)+\langle \dot{x}, C x \rangle.
\]
This implies
\[ \frac{d}{dt}\left(E_k-\frac{1}{2}\langle x, B x \rangle\right)=\langle \dot{x}, C x \rangle, \]
or equivalently
\begin{equation} \label{C-3-6}
d\left(E_k-\frac{1}{2}\langle x, B x \rangle\right)=\langle dx, C x \rangle.
\end{equation}
Since the left-hand side of \eqref{C-3-6} is in exact form, by taking differential both sides of \eqref{C-3-6}, one has 
\begin{equation} \label{C-3-7}
d \langle dx, C x \rangle = 0.
\end{equation}
We denote by $[x]_i$ be the $i$-th component of $x$. Then, the term $\langle dx, C x \rangle$ can be written as a component form:
\begin{equation} \label{C-3-8}
 \langle dx, C x \rangle  =  \sum_{j} [dx]_j [C x]_j =\sum_{i,j } [dx]_j [C]_{ji} [x]_i. 
\end{equation}
By \eqref{C-3-7}, \eqref{C-3-8} and skew-symmetry of $C$, one has 
\begin{align}
\begin{aligned} \label{C-3-9}
0 &= d \Big ( \sum_{i,j } [dx]_j [C]_{ji} [x]_i \Big) =   \sum_{i,j } [C]_{ji} [dx]_i \wedge [dx]_j  \\
&=   \sum_{i < j}  [C]_{ji}   [dx]_i \wedge [dx]_j +  \sum_{i > j}  [C]_{ji}   [dx]_i \wedge [dx]_j \\
& = \sum_{i < j}  [C]_{ji}   [dx]_i \wedge [dx]_j +  \sum_{j > i}  [C]_{ij}   [dx]_j \wedge [dx]_i \\
&  = \sum_{i < j}  ([C]_{ji} -   [C]_{ij} )  [dx]_i \wedge [dx]_j \\
& = -2\sum_{i < j}  [C]_{ij}   [dx]_i \wedge [dx]_j,
\end{aligned}
\end{align}
where we used the relations for wedge product:
\[ d[x]_i \wedge d[x]_i = 0, \quad d[x]_j \wedge d[x]_i = -d[x]_i \wedge d[x]_j, \quad 1 \leq i, j \leq d. \]
Thus, relation \eqref{C-3-9} implies
\[  [C]_{ij} = 0 \quad \mbox{for}~~i < j. \]
Then, by the skew symmetry of $C$, one obtains
\[ C \equiv 0 \]
and we obtain the desired estimate.
\end{proof}
Now, we return to the special situation in which only vertices with the same index interact each other:
\begin{equation} \label{C-8}
m\ddot{x}_\alpha^i= -\gamma {\dot x}^i_\alpha +\frac{\kappa}{N}\sum_{k=1}^N(x^k_\alpha-x^i_\alpha) +F_{\alpha, c}^i(=F_{\alpha, a}^i+F_{\alpha, c}^i).
\end{equation} 
We substitute the ansatz
\[
x_\alpha^i(t)=\bar{x}^i(t)+O^i(t)r_\alpha,
\]
into \eqref{C-8} to get
\begin{equation} \label{C-9}
m(\ddot{{\bar x}}^i+\ddot{O}^ir_\alpha)=-\gamma(\dot{{\bar x}}^i+\dot{O}^ir_\alpha)+\frac{\kappa}{N}\sum_{k=1}^N({\bar x}^k+O^kr_\alpha-{\bar x}^i-O^ir_\alpha)+F^i_{\alpha, c}.
\end{equation}
We use \eqref{C-7} to simplify \eqref{C-9} further as 
\begin{equation*} \label{C-9-1}
m\ddot{O}^i r_\alpha=-\gamma\dot{O}^ir_\alpha+\frac{\kappa}{N}\sum_{k=1}^N(O^k r_\alpha-O^i r_\alpha)+F^i_{\alpha,c}.
\end{equation*}
This gives the configuration matching force:
\begin{align}
\begin{aligned}  \label{C-10}
F_{\alpha, c}^i &=  \left(m\ddot{O}^i+\gamma \dot{O}^i-\frac{\kappa}{N}\sum_{k=1}^N(O^k-O^i)\right)r_\alpha  \\
&= \left(m\ddot{O}^i+\gamma \dot{O}^i-\frac{\kappa}{N}\sum_{k=1}^N(O^k-O^i)\right) (O_\alpha^i)^T r_\alpha^i,
\end{aligned}
\end{align}
where we used the relation:
\[ r_\alpha^i = O^i r_\alpha  \quad \mbox{or equivalently} \quad r_\alpha = (O^i)^T r_\alpha^i. \]
By comparing the ansatz in \eqref{C-3-1-1}, we set 
\begin{equation*} \label{C-10-1}
A^i:=\left(m\ddot{O}^i+\gamma \dot{O}^i-\frac{\kappa}{N}\sum_{k=1}^N(O^k-O^i)\right)(O^i)^T,
\end{equation*}
so that 
\begin{equation} \label{C-11}
F_{\alpha, c}^i = A^i r_\alpha^i =  A^i O^i r_\alpha, \qquad \forall~\alpha\in\Lambda.
\end{equation}
Then, it follows from Lemma \ref{L3.1} that $A^i$ is symmetric:
\[ (A^i)^T = A^i, \quad i = 1, \cdots, N.   \]
\begin{lemma}\label{L3.2}
Suppose the initial configuration lies in a general position such that
\begin{equation} \label{C-13}
\mbox{span} \{r_{\alpha}:~\alpha \in \Lambda \}  = \bbr^d,  
\end{equation}
and let $\{x_\alpha^i \}$ be a solution to \eqref{C-8}. Then, symmetric matrix $A^i$ is given as follows:
\[    A^i=-m \dot{O}^i(\dot{O}^i)^T-\frac{\kappa}{2N}\sum_{k=1}^N \Big[ O^k(O^i)^T+O^i(O^k)^T-2I_d \Big ]. \]
\end{lemma}
\begin{proof}

\noindent $\bullet$~Step A:  Let $O^i(t)$ be a smooth o$SO(d)$-valued function in $t$:
\[
O^i(O^i)^{T} =I_d.
\]
Then, we differentiate above relation with respect to $t$ successively twice to get 
\begin{equation} \label{C-13-1}
\dot{O}^i (O^i)^T+O^i (\dot{O}^i)^T= O_d \quad \mbox{and} \quad 
\ddot{O}^i (O^i)^T+2\dot{O}^i (\dot{O}^i)^T+O^i (\ddot{O}^i)^T= O_d.
\end{equation}

\vspace{0.5cm}

\noindent $\bullet$~Step B:~It follows from \eqref{C-10} and \eqref{C-11} that 
\[
\left(m\ddot{O}^i+\gamma \dot{O}^i-\frac{\kappa}{N}\sum_{k=1}^N(O^k-O^i)\right)r_\alpha=A^iO^ir_\alpha.
\]
This implies
\begin{equation} \label{C-13-1-1}
\left(m\ddot{O}^i+\gamma \dot{O}^i-\frac{\kappa}{N}\sum_{k=1}^N (O^k-O^i)-A^iO^i \right)r_\alpha= 0, \quad \forall~\alpha \in \Lambda.
\end{equation}
By the assumption \eqref{C-13}, relation \eqref{C-13-1-1} should hold for all vectors $y \in {\mathbb R}^d$, i.e., 
\begin{equation*} \label{C-13-1-2}
\left(m\ddot{O}^i+\gamma \dot{O}^i-\frac{\kappa}{N}\sum_{k=1}^N(O^k-O^i)-A^iO^i \right) y= 0, \quad \forall~y \in {\mathbb R}^d.
\end{equation*}
Thus, we have
\begin{equation} \label{C-13-2}
m\ddot{O}^i+\gamma \dot{O}^i-\frac{\kappa}{N}\sum_{k=1}^N(O^k-O^i)-A^iO^i= O_d, 
\end{equation}
i.e.,
\begin{equation} \label{C-13-3}
 A^i = m\ddot{O}^i  (O^i)^T  +\gamma \dot{O}^i  (O^i)^T  -\frac{\kappa}{N}\sum_{k=1}^N(O^k  (O^i)^T -I_d).
 \end{equation}
 We again take a transpose of \eqref{C-13-3} and use the relation $(A^i)^{T} = A^i$ to get 
\begin{equation} \label{C-14}
A^i = (A^i)^T =m  O^i (\ddot{O}^i)^T +\gamma O^i (\dot{O}^i)^T -\frac{\kappa}{N}\sum_{k=1}^N \Big(O^i (O^k)^T-I_d \Big).
\end{equation}
We add two relations \eqref{C-13-3} and \eqref{C-14} to see
\begin{align*}
\begin{aligned}
A^i &=\frac{m}{2} (\ddot{O}^i(O^i)^T+{O}^i(\ddot{O}^i)^T)+ \frac{\gamma}{2} (\dot{O}^i (O^i)^T+  O^i(\dot{O}^i)^T) \\
&-\frac{\kappa}{2N}\sum_{k=1}^N(O^k(O^i)^T+O^i(O^k)^T-2I_d).
\end{aligned}
\end{align*}
Now we use the relations \eqref{C-13-1} in Step A to simply the above relation as
\begin{align}\label{C-15}
A^i=-m \dot{O}^i(\dot{O}^i)^T-\frac{\kappa}{2N}\sum_{k=1}^N \Big(O^k(O^i)^T+O^i(O^k)^T-2I_d \Big).
\end{align}
\end{proof}


Now we are ready to derive the dynamics for $O^i$. We combine \eqref{C-13-2} and \eqref{C-15} to get 
\begin{align*}
\begin{aligned}
& -m \dot{O}^i(\dot{O}^i)^T-\frac{\kappa}{2N}\sum_{k=1}^N(O^k(O^i)^T+O^i(O^k)^T-2I_d) \\
& \hspace{4cm} =m\ddot{O}^i(O^i)^T+\gamma \dot{O}^i (O^i)^T-\frac{\kappa}{N}\sum_{k=1}^N(O^k(O^i)^T-I_d),
\end{aligned}
\end{align*}
After rearrangement, one has the second-order Lohe matrix model introduced in \cite{H-K}:
\begin{align*}
m(\ddot{O}^i(O^i)^T+\dot{O}^i(\dot{O}^i)^T)+\gamma \dot{O}^i (O^i)^T=\frac{\kappa}{2N}\sum_{k=1}^N \Big(O^k(O^i)^T-O^i(O^k)^T \Big).
\end{align*}
Note that the above system was set up in the aforementioned reference without any derivations, whereas, we have derived the second-order matrix model starting from reasonable physical arguments. \newline

In the following section, we provide generalizations for the consensus algorithms for the ensemble of similar polytopes and mixed ensemble consisting of distinct types of polytopes.

\section{Extension and analytical results} \label{sec:4}
\setcounter{equation}{0}
In this section, we present two possible extensions of the model \eqref{A-1} for  the shape matchings in the ensembles of similar polytopes and heterogeneous ones, and provide analytical results for the complete shape matching on the ensemble of congruent polytopes. 

\subsection{Extensions to similar and heterogeneous polytopes} \label{sec:4.1}

 In this subsection, we provide straightforward extensions for the model \eqref{A-1} for two ensembles of similar polytopes with the same geometric shapes and mixed of them with different shapes. 
\subsubsection{Similar polytope ensemble} \label{sec:4.1.1}
Consider an ensemble of similar polytopes with the same geometric shape. Then, due to non-congruence between polytopes, the vertices of polytope will not be matched exactly, so we instead a design a dynamical system for similar polytopes leading to the regularly placed patterns, e.g., concentric circles with the same center for a family of circles. 

For a referenced family of largest congruent polytopes $\{X^i=\{x_\alpha^i\}_{\alpha\in\Lambda} \}_{i=1}^{N},$ we set the decomposition of each vertex point $x^i_\alpha$ as
\[
x_\alpha^i(t)=\bar{x}^i(t)+O^i(t)r_\alpha
\]
for some vectors $\{r_\alpha\}_{\alpha\in\Lambda}$. For other family of similar polytopes $\{Y_\alpha^i\}_{i}$, there exists a positive ratio $s^i > 0$ such that 
\[
y_\alpha^i(t)=\bar{y}^i(t)+s^i O^i(t) r_\alpha, \quad  i = 1, \cdots, N. 
\]
From the same arguments in Sections \ref{sec:3.1} and \ref{sec:3.2}, we can obtain the following dynamical system for $\{\bar{y}^i,O^i\}$:
\begin{align}\label{F-1}
\begin{cases}
\displaystyle m\ddot{{\bar y}}^i = -\gamma\dot{{\bar y}}^i + \frac{\kappa}{N}\sum_{k=1}^N({\bar y}^k- {\bar y}^i), \quad r^i_\alpha(t)= s^i O^i(t) r_\alpha, \\
\displaystyle m\Big(\ddot{O}^i(O^i)^T+ \dot{O}^i(\dot{O}^i)^T\Big) = -\gamma\dot{O}^i(O^i)^T + \frac{\kappa}{2N}\sum_{k=1}^N\frac{s^k}{s^i}\big(O^k(O^i)^T-O^i(O^k)^T\big),\\
\displaystyle ({\bar y}^i, \dot{\bar y}^i) \Big|_{t= 0+} = ({\bar y}^{i0}, \dot{\bar y}^{i0}) \quad \mbox{and} \quad 
({\bar O}^i, \dot{\bar O}^i) \Big|_{t= 0+} = ({\bar O}^{i0}, \dot{\bar O}^{i0}), \\
O^{i0}\in\mathrm{SO}(d),\quad \dot{O}^{i0} (O^{i0})^T+O^{i0} (\dot{O}^{i0})^T= O_d,\quad, \forall~i = 1, \cdots, N.
\end{cases}
\end{align}
Note that for the case
 \[  s^i = s^k, \quad 1 \leq i, k \leq N, \]
 system \eqref{F-1} reduces to system \eqref{A-1}. The factor $\frac{s^k}{s^i}$ in the coupling term of $\eqref{F-1}_2$ may not be symmetric with respect to $i$ and $k$. Hence the analysis employed in \cite{H-K} may not be applied in our situation directly. Thus, we do not know whether subsystem $\eqref{F-1}_2$ will exhibit emergent dynamics or not. As far as the authors know, even for the first-order Lohe matrix model on the non-symmetric network, emergent behaviors were not studied in previous literature. Thus, we only study emergent behaviors via  numerical simulations in Section \ref{sec:5}.

\subsubsection{Heterogeneous polytopes ensemble} \label{sec:4.1.2}
In this subsection, we discuss a possible extension of \eqref{A-1} for the aggregation of polytopes with the same shape in an ensemble of  heterogeneous polytopes. To avoid collapse between distinct types of polytopes, we use a quadratic potential \cite{P-K-H} between different type of polytopes:
\begin{equation*} \label{F-1-1}
V(x,y)=\frac{\kappa_2}{2}(\|x-y\|-L)^2.
\end{equation*}
To simplify following discussions, we assume that there are two types of polytopes under considerations. We assume that there are $N_1$ and $N_2$ for type 1 and type 2 respectively. We also denote the coordinate of type 1 and type 2 as follows:
\begin{align*}
x^i_\alpha=\bar{x}^i+U^i r_\alpha\quad\forall i\in \{1, 2, \cdots, N_1\},\quad \alpha\in \Lambda_1,\\
y^i_\beta=\bar{y}^i+V^i \tilde{r}_\beta\quad\forall j\in \{1, 2, \cdots, N_2\}, \quad \beta\in \Lambda_2. 
\end{align*}
Then we propose the following coupled dynamics for $(\bar{x}^i, U^i)$ and $(\bar{y}^j, V^j)$:
\begin{equation} \label{F-2}
\begin{cases}
\displaystyle m\ddot{{\bar x}}^i+\gamma\dot{{\bar x}}^i-\displaystyle\frac{\kappa_1}{N_1}\sum_{k=1}^{N_1}({\bar x}^k- {\bar x}^i)-\frac{\kappa_2}{N_2}\sum_{l=1}^{N_2}(\|\bar{y}^l-\bar{x}^i\|-L)\frac{\bar{y}^l-\bar{x}^i}{\|\bar{y}^l-\bar{x}^i\|}=0,  \\
\displaystyle m\ddot{{\bar y}}^j+\gamma\dot{{\bar y}}^j-\displaystyle\frac{\kappa_1}{N_2}\sum_{l=1}^{N_2}({\bar y}^l- {\bar y}^j)-\frac{\kappa_2}{N_1}\sum_{k=1}^{N_1}(\|\bar{x}^k-\bar{y}^j\|-L)\frac{\bar{x}^k-\bar{y}^j}{\|\bar{x}^k-\bar{y}^j\|}=0,  \\
\displaystyle m\Big(\ddot{U}^i(U^i)^T+\dot{U}^i(\dot{U}^i)^T\Big)+\gamma\dot{U}^i(U^i)^T-\displaystyle\frac{\kappa_1}{2N_1}\sum_{k=1}^{N_1}\big(U^k(U^i)^T-U^i(U^k)^T\big)=0,\\
\displaystyle m\Big(\ddot{V}^j(V^j)^T+\dot{V}^j(\dot{V}^j)^T\Big)+\gamma\dot{V}^j(V^j)^T-\displaystyle\frac{\kappa_1}{2N_2}\sum_{l=1}^{N_2}\big(V^l(V^j)^T-V^j(V^l)^T\big)=0,\\
\displaystyle ({\bar x}^i, \dot{\bar x}^i) \Big|_{t= 0+} = ({\bar x}^{i0}, \dot{\bar x}^{i0}), \quad 
({\bar U}^i, \dot{\bar U}^i) \Big|_{t= 0+} = ({\bar U}^{i0}, \dot{\bar U}^{i0}),\\
\displaystyle ({\bar y}^j, \dot{\bar y}^j) \Big|_{t= 0+} = ({\bar y}^{j0}, \dot{\bar y}^{j0}),\quad 
({\bar V}^j, \dot{\bar V}^j) \Big|_{t= 0+} = ({\bar V}^{j0}, \dot{\bar V}^{j0}),\\
 U^{i0},~V^{j0}\in\mathrm{SO}(d),\quad \dot{U}^{i0} (U^{i0})^T+U^{i0} (\dot{U}^{i0})^T= O_d,\quad \dot{V}^{j0} (V^{j0})^T+V^{j0} (\dot{V}^{j0})^T= O_d\\
\forall~i \in \{1,2, \cdots, N_1\},\quad  j\in \{1, 2, \cdots, N_2\},\quad L\in\bbr^+.
\end{cases}
\end{equation}
where $\kappa_1$ and $\kappa_2$ are nonnegative constants. \newline

Note that compared to \eqref{A-1} for a homogeneous ensemble, new forcing terms appearing in the R.H.S. of \eqref{F-2} are 
\begin{align*}
\begin{aligned}
& (\|\bar{y}^l-\bar{x}^i\|-L)\frac{\bar{y}^l-\bar{x}^i}{\|\bar{y}^l-\bar{x}^i\|} = -\nabla_{x^i} V(x^i,y^l), \\
& (\|\bar{x}^k-\bar{y}^j\|-L)\frac{\bar{x}^k-\bar{y}^j}{\|\bar{x}^k-\bar{y}^j\|} = -\nabla_{y^j} V(x^k,y^j). 
\end{aligned}
\end{align*}
In the absence of repulsive forcing with $\kappa_2 = 0$, system \eqref{F-2} becomes the juxtaposition of two models \eqref{A-1} for $(\bar{x}^i, U^i)$ and $(\bar{y}^j, V^j)$. 
\subsubsection{The second-order dynamics} \label{sec:4.1.3}
Consider the Cauchy problem to the second-order dynamics:
\begin{align}\label{D-1}
\begin{cases}
\displaystyle m\ddot{{\bar x}}^i = -\gamma\dot{{\bar x}}^i + \frac{\kappa}{N}\sum_{k=1}^N({\bar x}^k- {\bar x}^i), \quad r^i_\alpha(t)=O^i(t) r_\alpha, \\
\displaystyle m\Big(\ddot{O}^i(O^i)^T+ \dot{O}^i(\dot{O}^i)^T\Big) = -\gamma\dot{O}^i(O^i)^T + \frac{\kappa}{2N}\sum_{k=1}^N\big(O^k(O^i)^T-O^i(O^k)^T\big),\\
\displaystyle ({\bar x}^i, \dot{\bar x}^i) \Big|_{t= 0+} = ({\bar x}^{i0}, \dot{\bar x}^{i0}) \quad \mbox{and} \quad 
({\bar O}^i, \dot{\bar O}^i) \Big|_{t= 0+} = ({\bar O}^{i0}, \dot{\bar O}^{i0}), \\
O^{i0}\in\mathrm{SO}(d),\quad \dot{O}^{i0} (O^{i0})^T+O^{i0} (\dot{O}^{i0})^T= O_d\quad \forall~i = 1, \cdots, N.
\end{cases}
\end{align}
Note that the center-of-mass dynamics and displacement dynamics are completely decoupled, so we can use results summarized Propositions \ref{P2.1} and \ref{P2.2} in Section \ref{sec:2}. 
\begin{theorem} \label{T4.1}
\emph{(Complete shape matching)}
Let $\{X^i(t)\}_{i=1}^N=\{\{x_\alpha^i(t)\}_{\alpha\in \Lambda}\}_{i=1}^N$ be a solution to \eqref{D-1}. Then, one has the complete aggregation:
\begin{align*}
\lim_{t \to \infty} \|x_\alpha^i(t)-x_\alpha^j(t) \| = 0,
\end{align*}
for all $\alpha\in \Lambda$ and $i, j\in \{1, 2, \cdots, N\}$.
\end{theorem}
\begin{proof}
It follows from Propositions \ref{P2.1} and \ref{P2.2} that 
\begin{equation} \label{D-1-1}
\lim_{t \to \infty} \|{\bar x}^i(t)-{\bar x}^j(t)\| = 0,\quad \lim_{t \to \infty}\|O^i(t) -O^j(t) \|_F = 0,
\end{equation}
for all $i,j=1, 2, \cdots, N$. Then it follows from \eqref{dec} and \eqref{D-1-1} that 
\begin{align*}
\begin{aligned}
\|x_\alpha^i-x_\alpha^j \|&=\|( {\bar x}^i+O^ir_\alpha)-( {\bar x}^j+O^jr_\alpha) \|\leq \|{\bar x}^i- {\bar x}^j \|+ \|(O^i-O^j)r_\alpha \|\\
&\leq  \| {\bar x}^i- {\bar x}^j \|+\|O^i-O^j \|_{F}\cdot \|r_\alpha \|\to0\quad\mbox{as}\quad t\to\infty
\end{aligned}
\end{align*}
which yields the desired result. 
\end{proof}

\subsubsection{The first-order dynamics} \label{sec:4.1.4}
We recall the Cauchy problem to the first-order model:
\begin{align}\label{D-2}
\begin{cases}
\displaystyle \dot{{\bar x}}^i = \frac{\kappa}{N}\sum_{k=1}^N({\bar x}^k- {\bar x}^i), \quad r^i_\alpha(t)=O^i(t) r_\alpha, \\
\displaystyle \dot{O}^i(O^i)^T = \frac{\kappa}{2N}\sum_{k=1}^N\big(O^k(O^i)^T-O^i(O^k)^T\big)=0,\\
\displaystyle {\bar x}^i |_{t= 0+} = {\bar x}^{i0} \quad \mbox{and} \quad O^i |_{t= 0+} = O^{i0} \in SO(d), \quad  \forall~i = 1, \cdots, N.
\end{cases}
\end{align}
Note that first-order dynamics \eqref{D-2} can be obtained in a formal limit:
\begin{align*}
\frac{m}{\gamma}\to0,\quad \frac{\kappa}{\gamma} =: \tilde{\kappa}.
\end{align*}
System \eqref{D-2} is a combination of the linear consensus model and the Lohe matrix model. Dynamics of each center of mass follows the linear consensus model, and the each rotational motion follows the Lohe matrix model. So we have the new approaches to Lohe matrix model.

\begin{theorem} \label{T4.2}
\emph{(Complete shape matching)}
Let $\{X^i(t)\}_{i=1}^N=\{\{x_\alpha^i(t)\}_{\alpha\in \Lambda}\}_{i=1}^N$ be a solution of the first-order linear consensus model \eqref{D-2} satisfying the relations:
\[  \kappa > 0, \quad  \max_{i,j} \|O^{i0} - O^{j0} \|_{F} < 1. \]   
Then, one has 
\begin{align*}
\lim_{t \to \infty} \|x_\alpha^i(t) -x_\alpha^j(t) \| = 0,
\end{align*}
for all $\alpha\in \Lambda$ and $i, j\in \{1, 2, \cdots, N\}$.
\end{theorem}
\begin{proof} The proof is exactly the same as in the proof of Theorem \ref{T4.1}. Hence we omit its details. 
\end{proof}
The rigorous emergent dynamics to \eqref{F-1} and \eqref{F-2}  of similar polytopes and heterogeneous polytopes will not be treated in this paper. 

\section{Numeric simulations} \label{sec:5}
\setcounter{equation}{0}
In this subsection, we provide several numerical examples for the complete shape matchings fo polygons and simplexes. 

\subsection{Congruent triangles} \label{sec:5.1}
In this subsection, we provide several numerical examples with the ensemble of congruent triangles and similar triangles. In Figure \ref{F1}, we provide a series of snapshots at $t= 0, 3, 6$ and $12$. 
\begin{figure}[h]
\centering
\mbox{
\subfigure[~$t=0$]{\includegraphics[scale = 0.45]{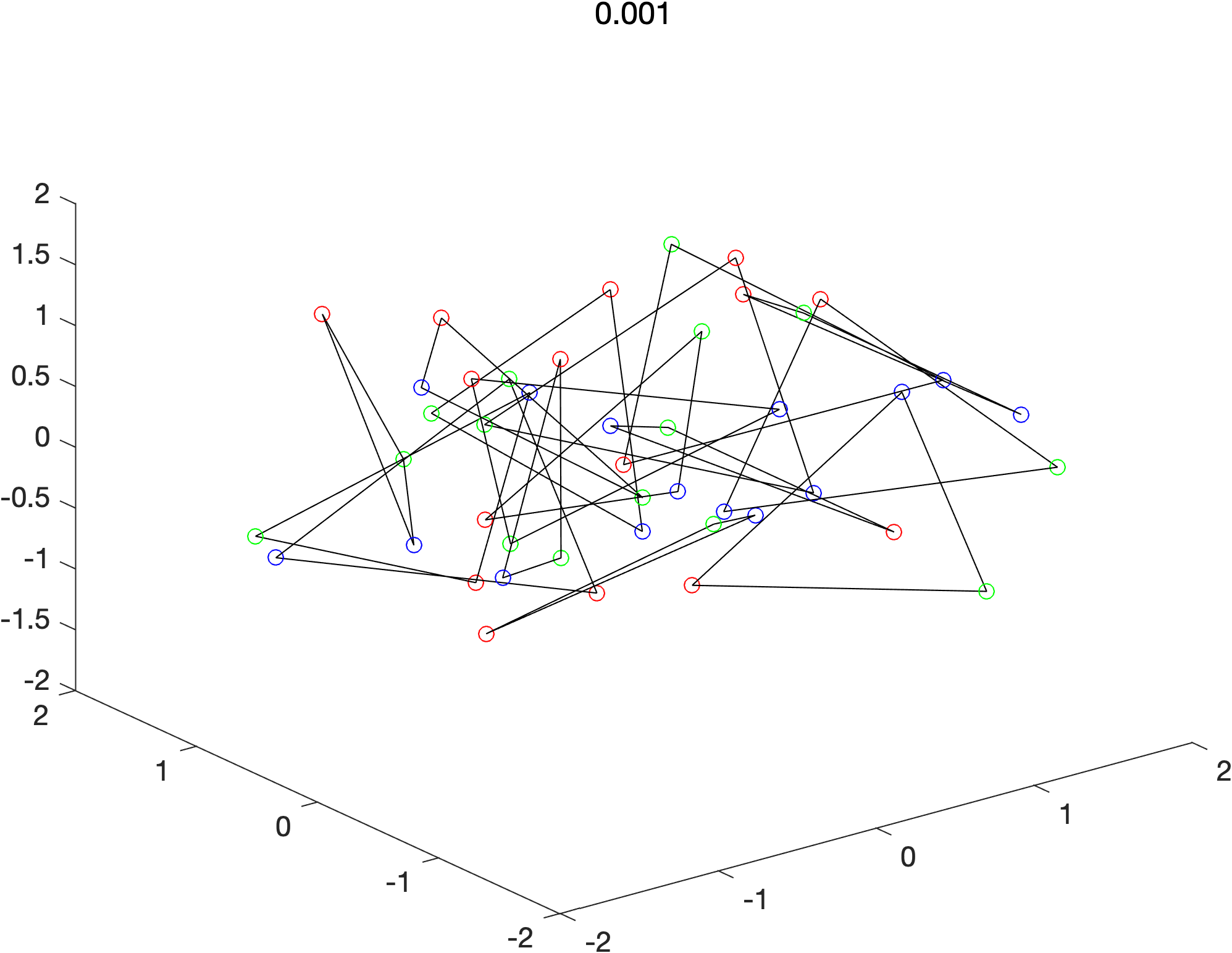}}
\subfigure[~$t=3$]{\includegraphics[scale = 0.45]{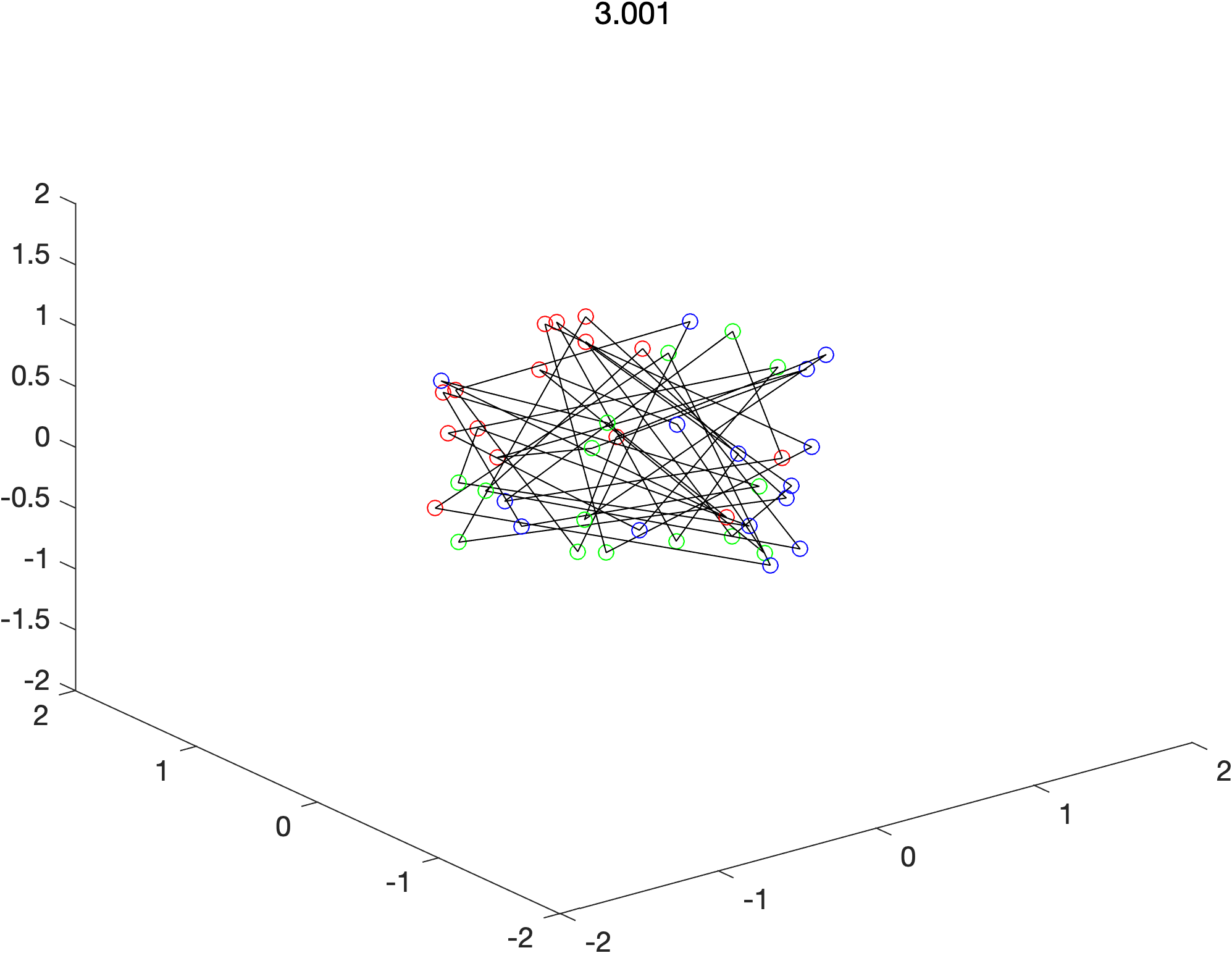}}
}
\end{figure}
\begin{figure}[h]
\centering
\mbox{
\subfigure[~$t=6$]{\includegraphics[scale = 0.45]{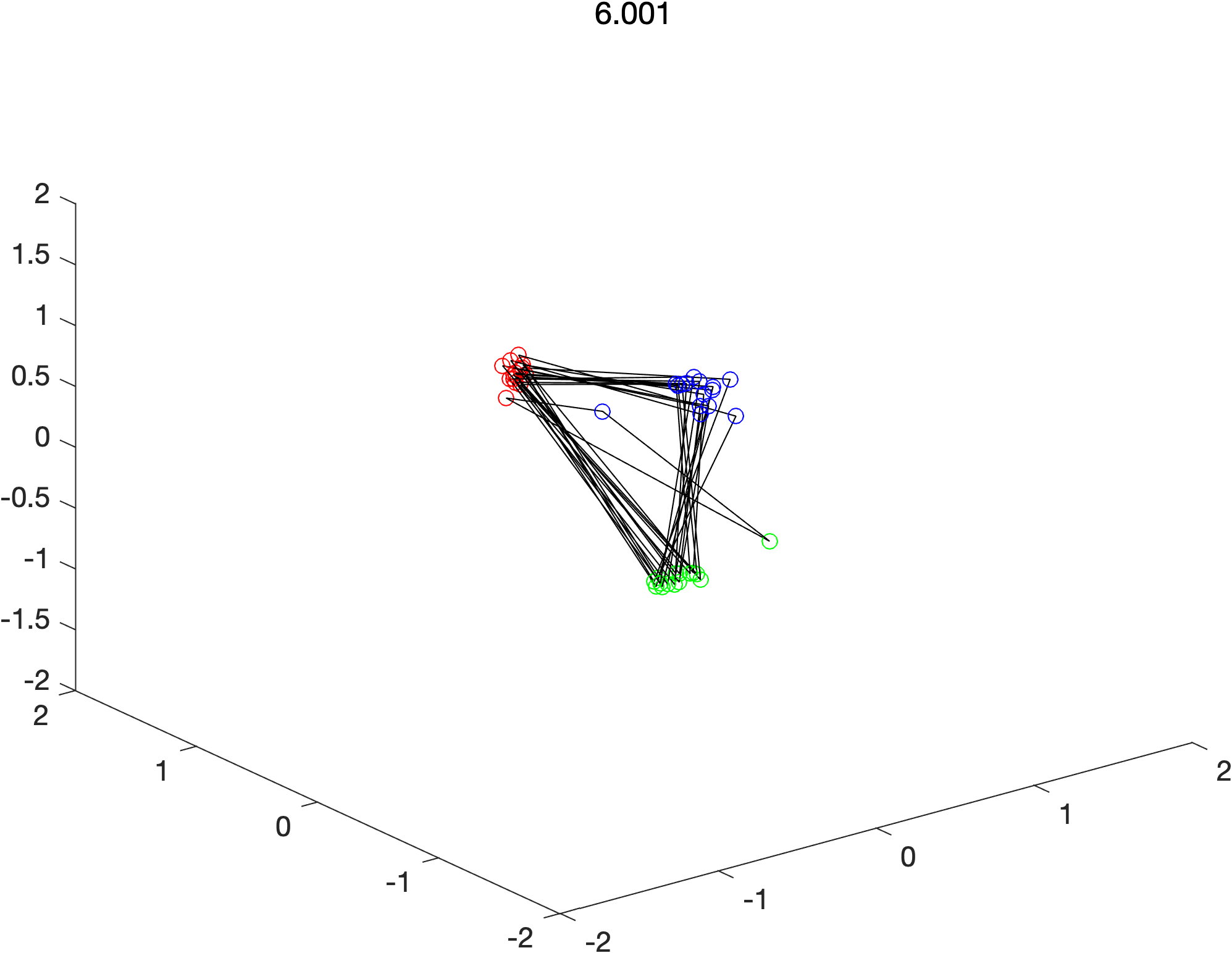}}
\subfigure[~$t=12$]{\includegraphics[scale = 0.45]{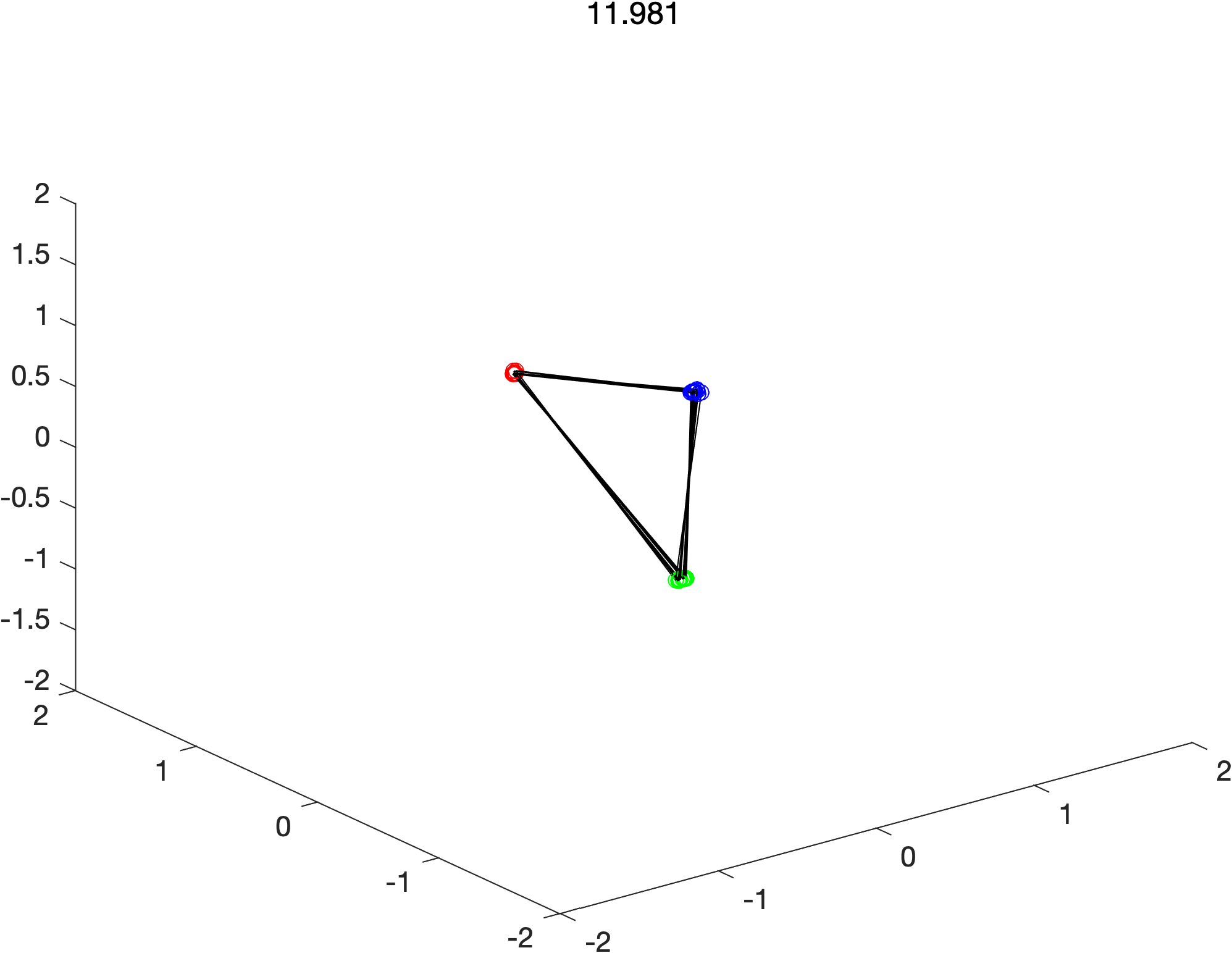}}
}
\caption{Complete shape matching of congruent triangles}
\label{F1}
\end{figure}

%
%
%
%

\vspace{0.5cm}

\subsection{Similar triangles} \label{sec:5.2}
In this subsection, we provide several simulations for the shape matchings of similar polytopes as in Figure \ref{F2} beginning from random initial configuration. 
\begin{figure}[h]
\centering
\mbox{
\subfigure[~$t=0$]{\includegraphics[scale = 0.45]{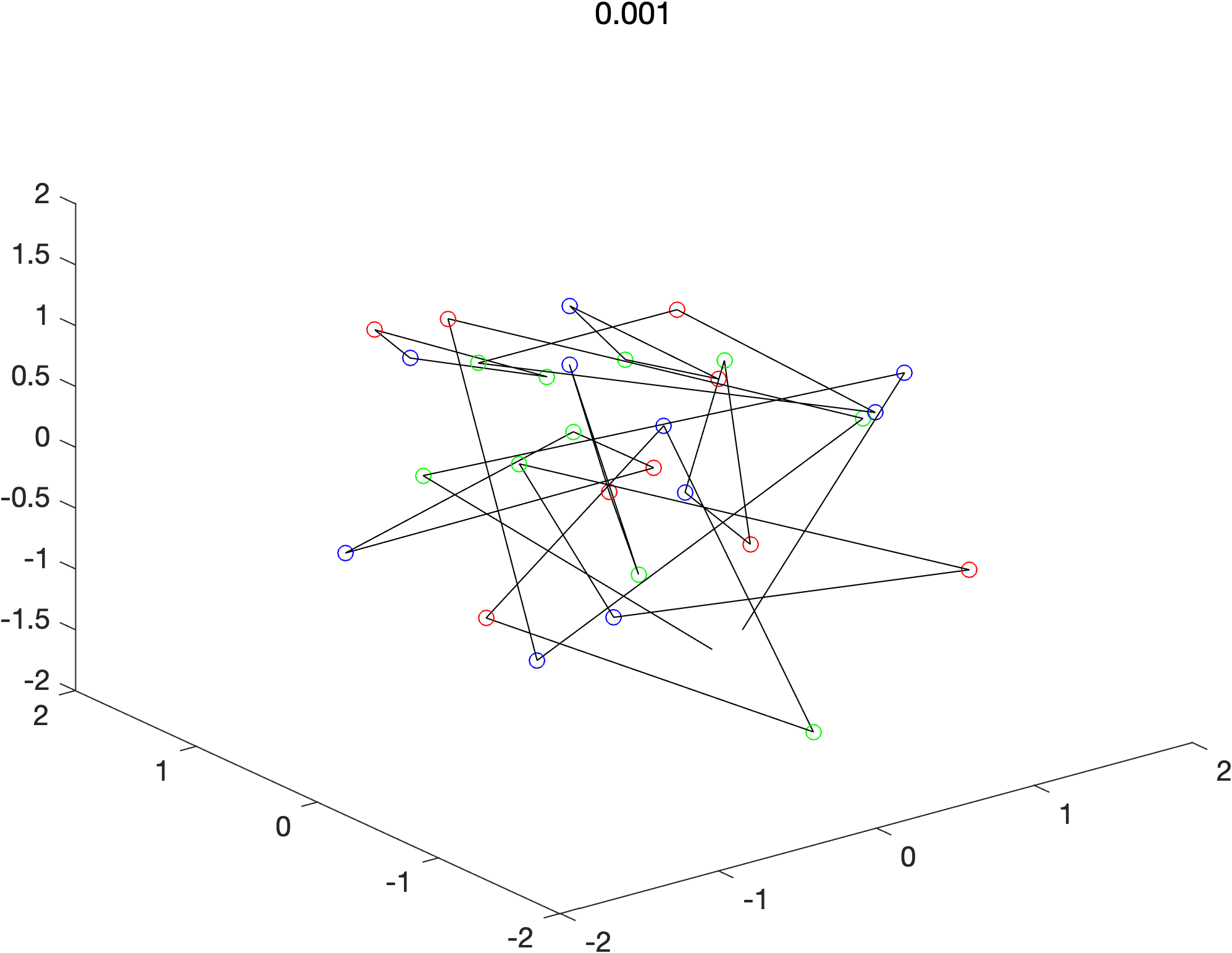}}
\subfigure[~$t=3$]{\includegraphics[scale = 0.45]{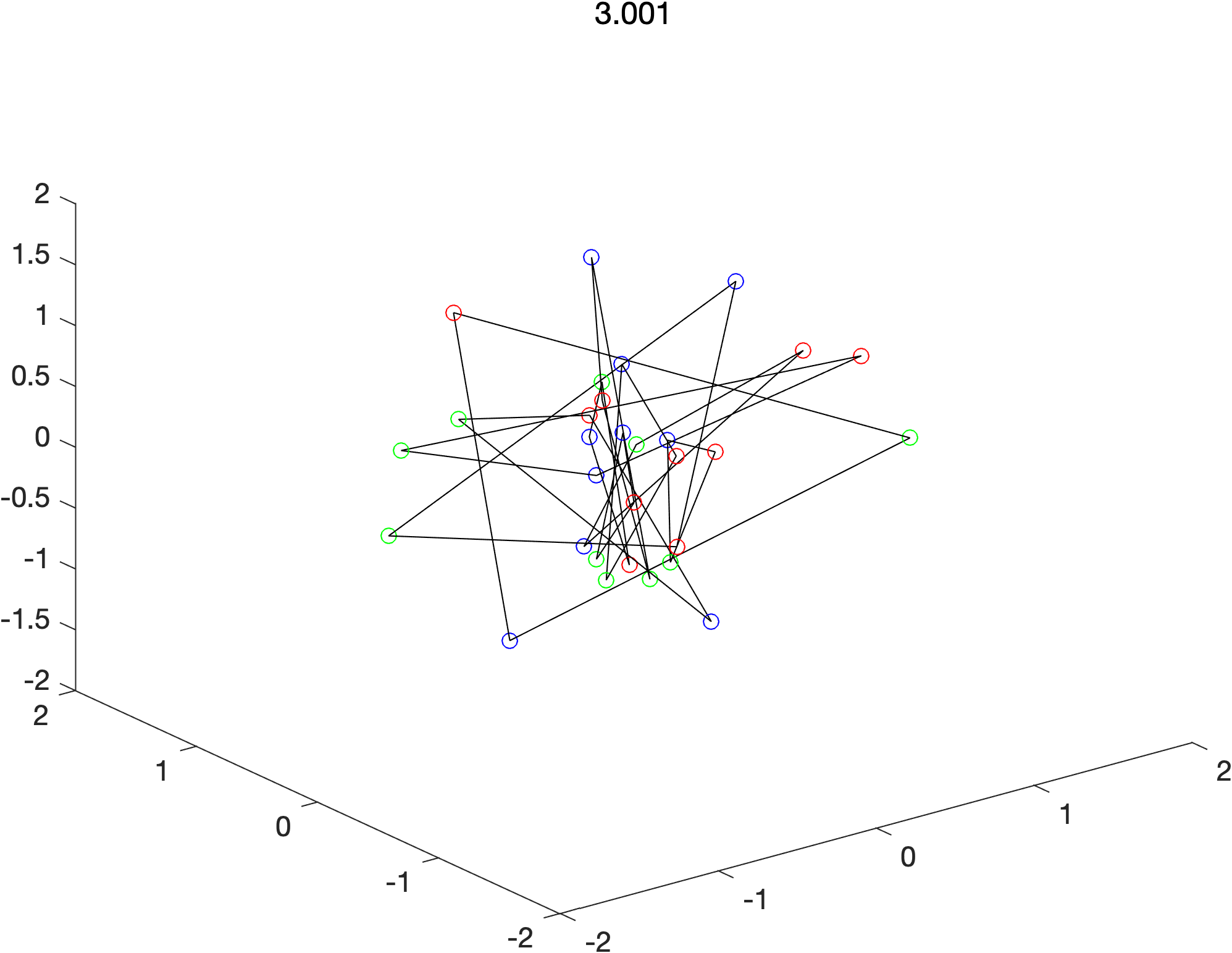}}
}
\end{figure}
\begin{figure}[h]
\centering
\mbox{
\subfigure[~$t=6$]{\includegraphics[scale = 0.45]{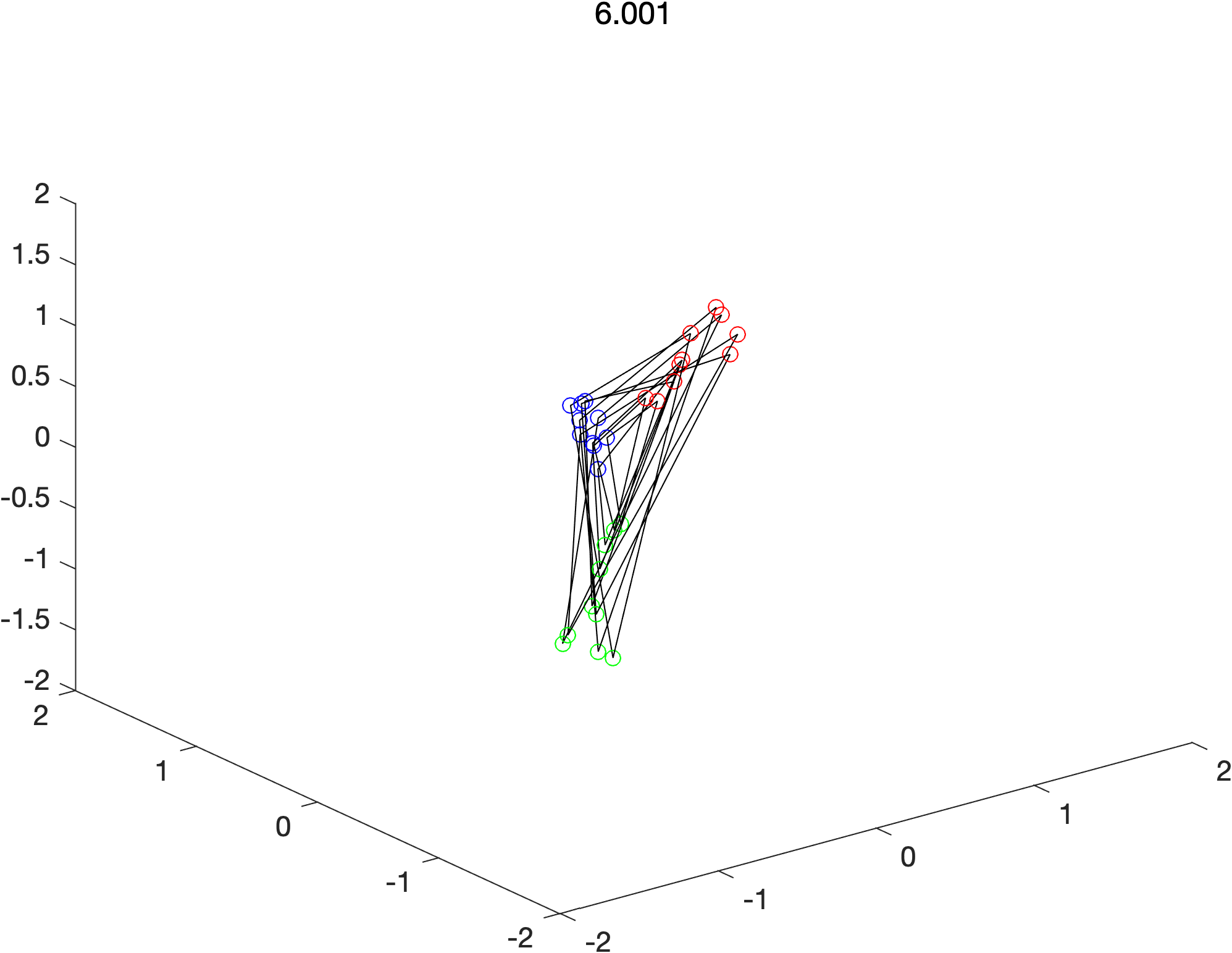}}
\subfigure[~$t=12$]{\includegraphics[scale = 0.45]{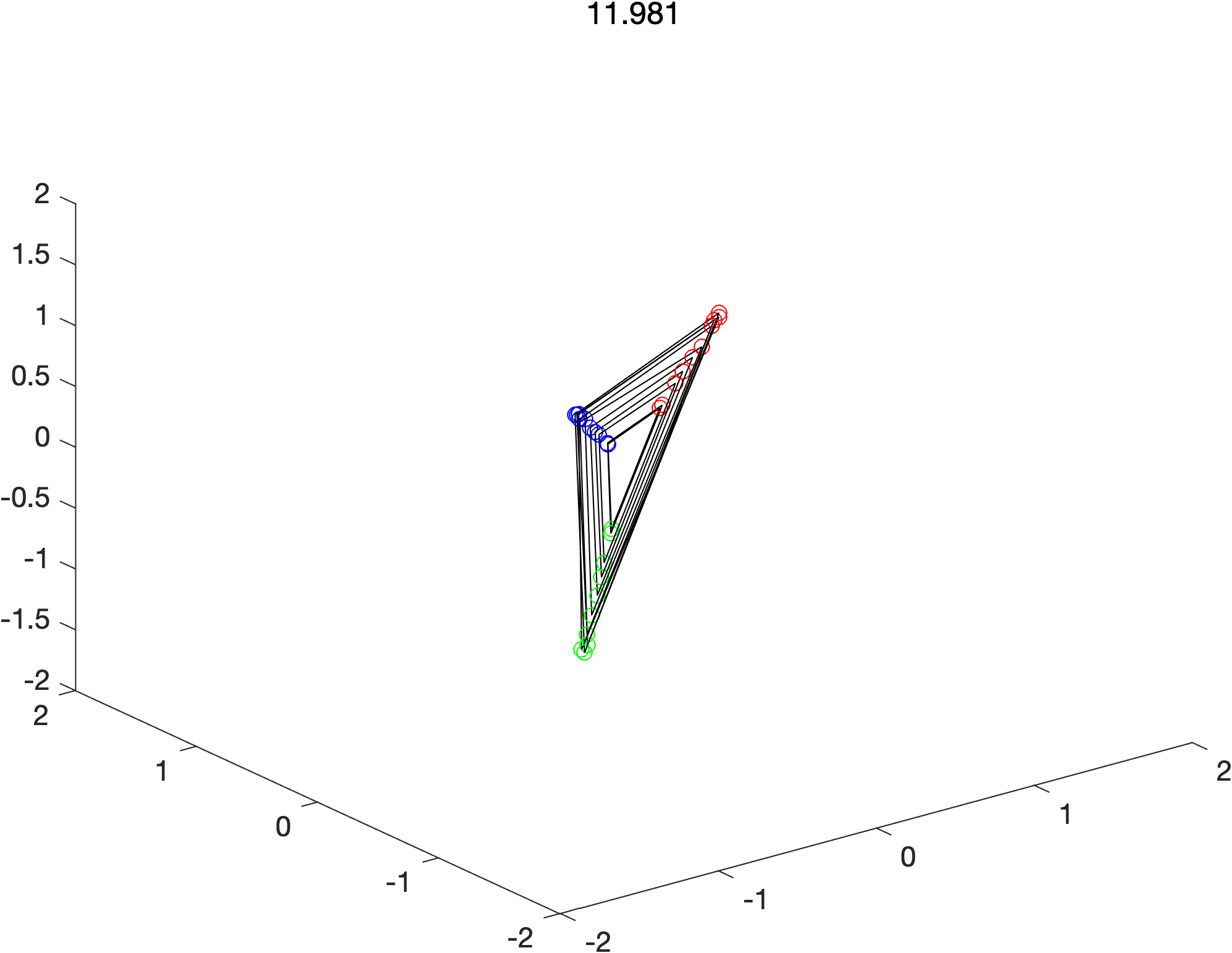}}
}
\caption{Complete shape matchings for similar triangles}
\label{F2}
\end{figure}


\subsection{Heterogeneous polytopes} \label{sec:5.3}
In this subsection, we emergence local shape matchings for the mixed ensemble of congruent triangles and congruent tetrahedrons. Beginning from random initial configuration, we observe how congruent triangles and congruent tetrahedrons are first separated and then each subensemble tends to complete shape matchings. In Figure \ref{F3}, we provide a series of snapshots at $t= 0, 3, 6$ and $12$. 
\begin{figure}[h]
\centering
\mbox{
\subfigure[~$t=0$]{\includegraphics[scale = 0.45]{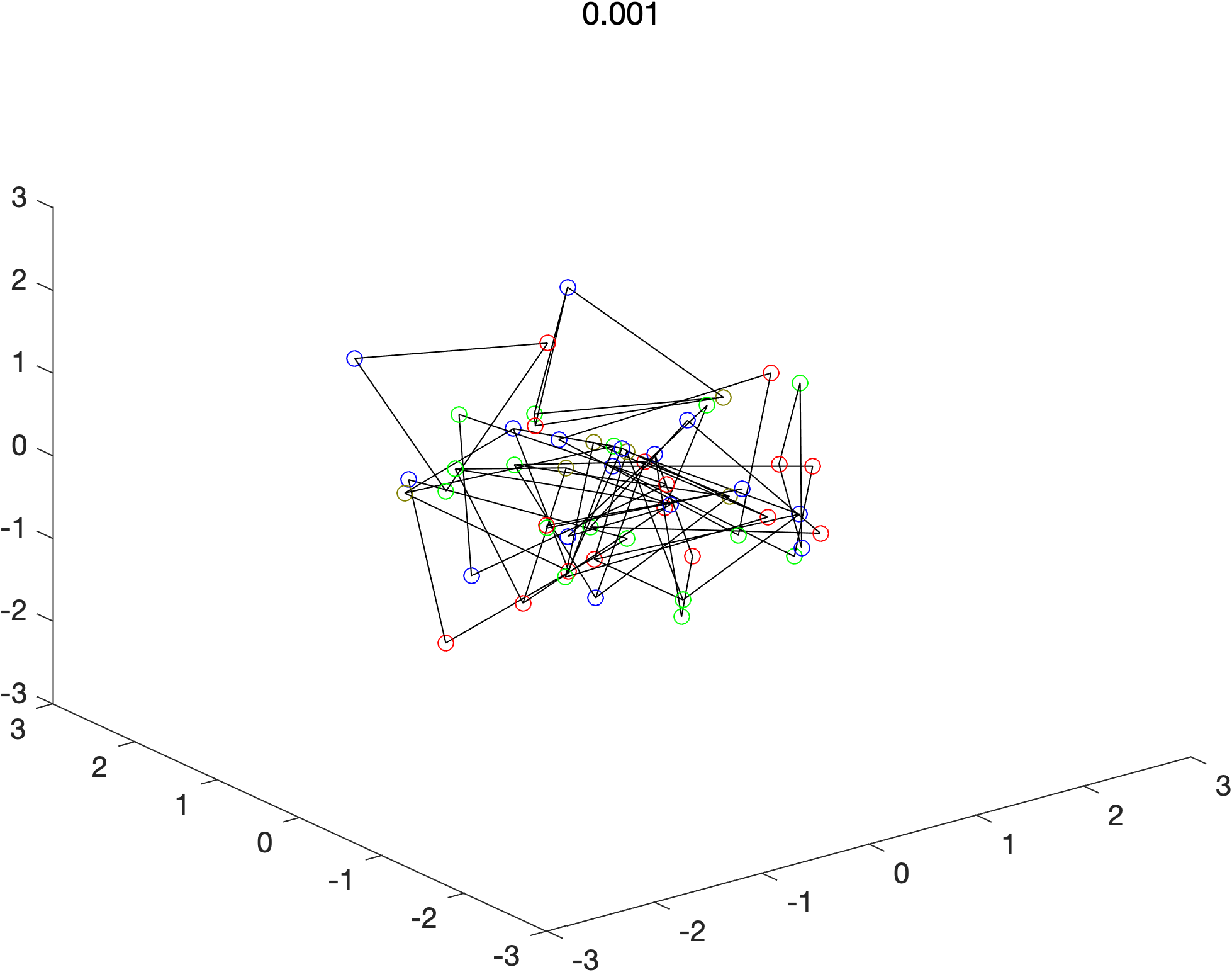}}
\subfigure[~$t=3$]{\includegraphics[scale = 0.45]{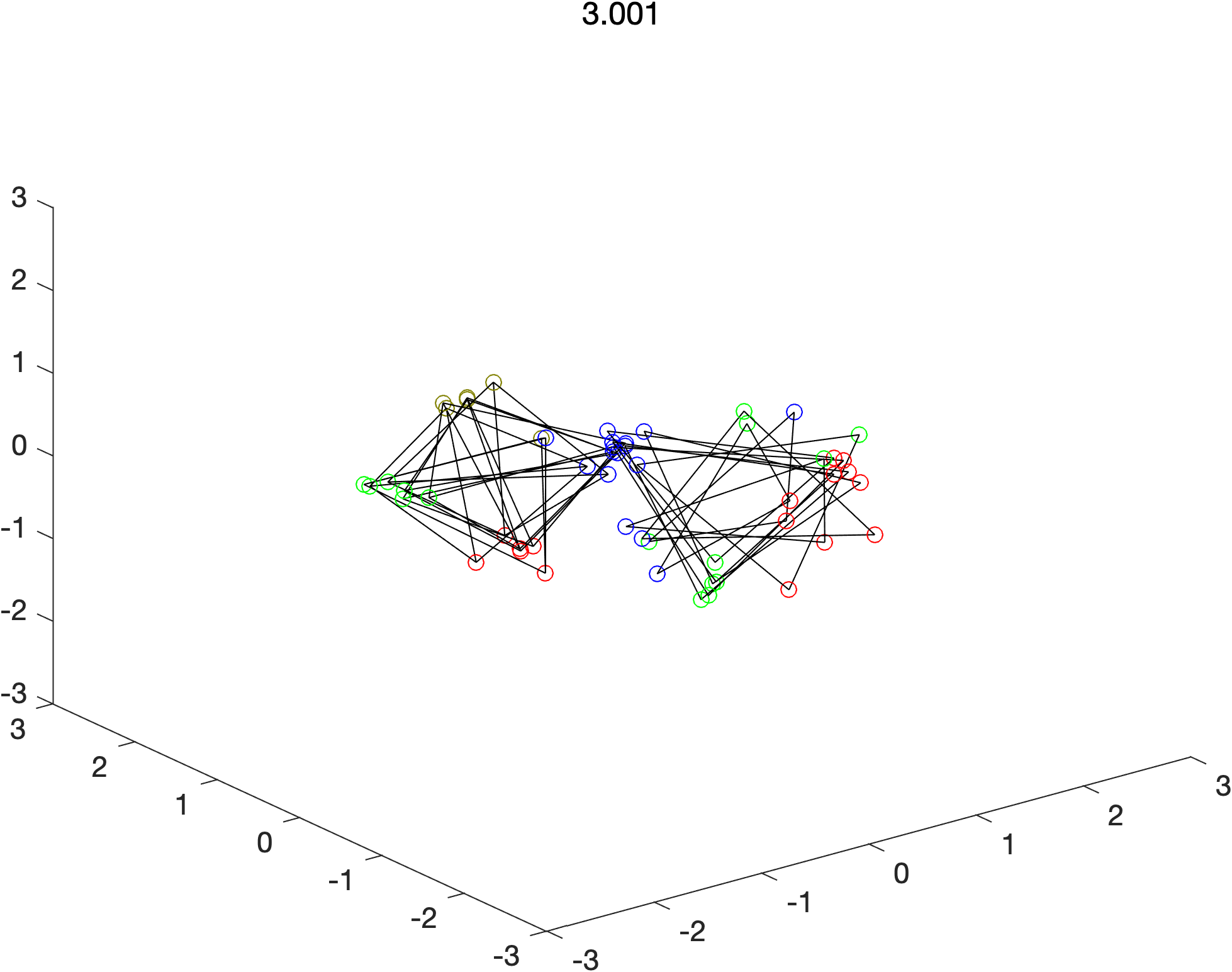}}
}
\end{figure}
\begin{figure}[h]
\centering
\mbox{
\subfigure[~$t=6$]{\includegraphics[scale = 0.45]{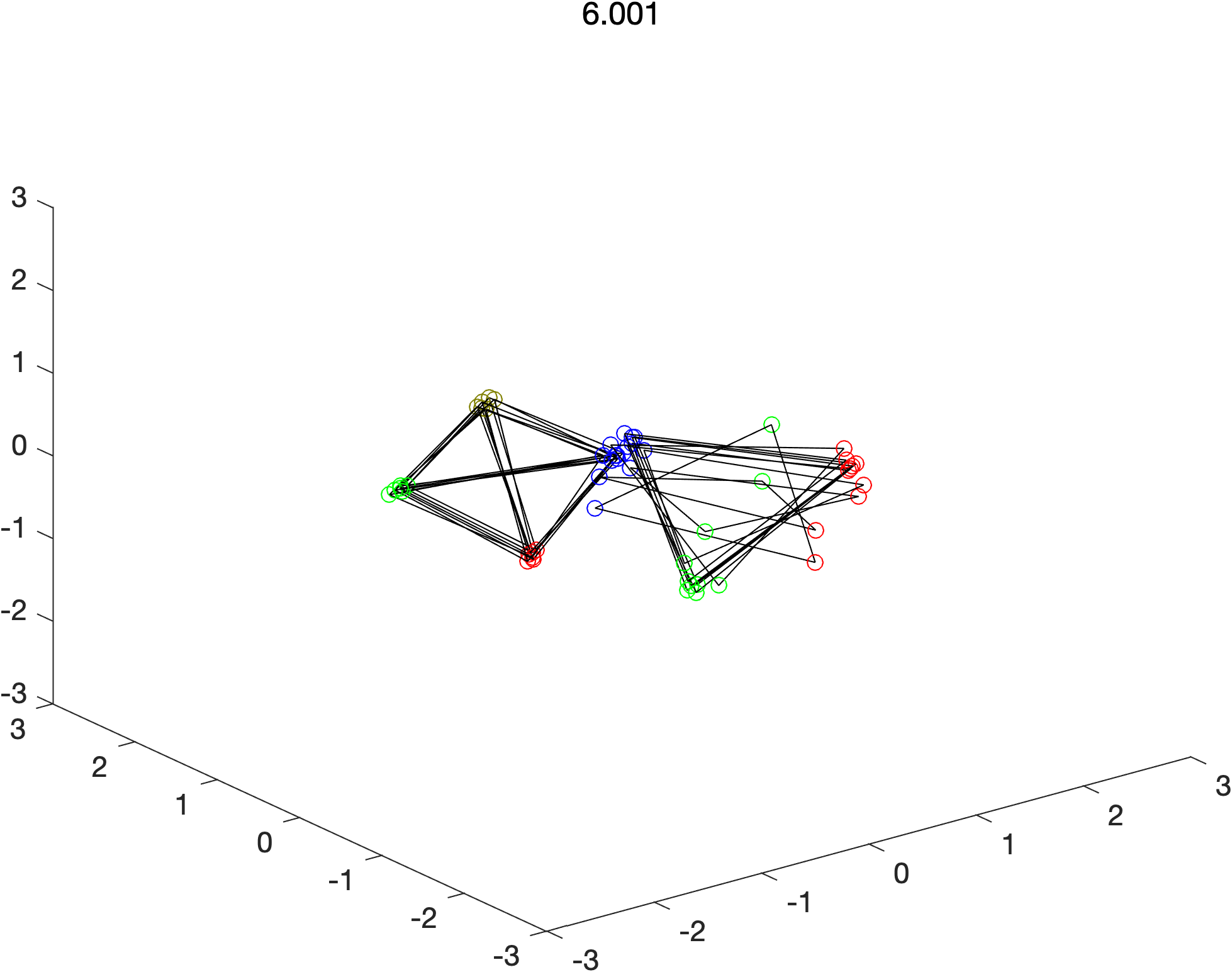}}
\subfigure[~$t=12$]{\includegraphics[scale = 0.45]{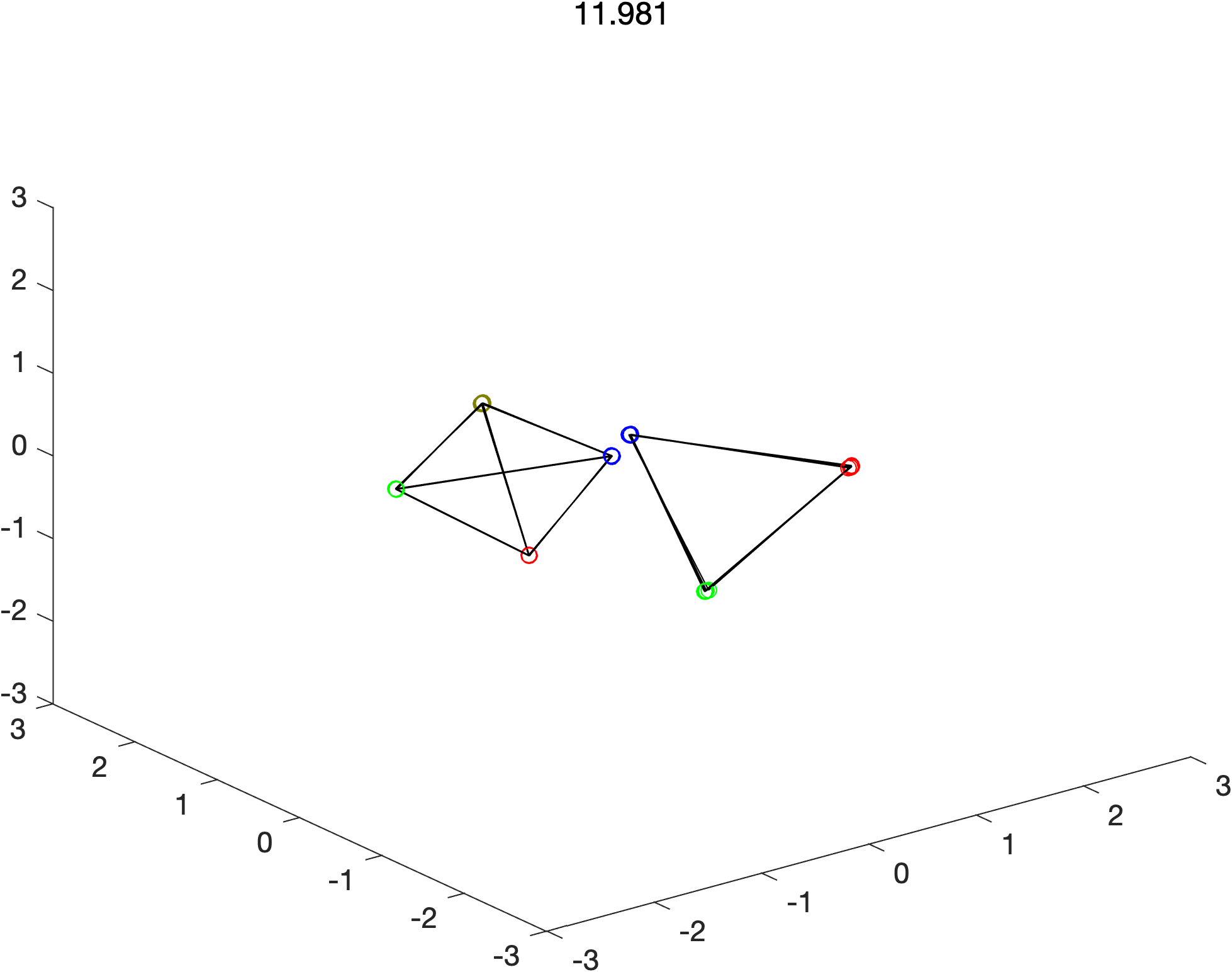}}
}
\caption{Complete shape matching for congruent tetrahedrons and congruent triangles}
\label{F3}
\end{figure}

\section{Conclusion} \label{sec:6}
\setcounter{equation}{0}
In this paper, we have introduced a dynamical systems approach for shape matchings of the ensemble of polytopes along rigid-body motions.
 As aforementioned in Introduction, so far, a dynamical systems approach is available only for the aggregation of point particles so that internal structure does not matter in the dynamics. In this work, we considered an ensemble consisting of polytopes such as polygons and simplexes. The rigid-body motions without refections can be decomposed as a direct sum of translation and rotations i.e., translation and rotations will be represented  by a vector in $\bbr^d$ and a matrix in $SO(d)$. Based on physical argument on the forces acting on the vertex of polytope, we showed that the center-of-mass will follow a system of second-order equation with linear damping and distributed consensus coupling, whereas the rotation matrix follows the Lohe matrix model on $SO(d)^N$. As far as the authors know, this is the first time for the Lohe matrix model to be derived based on a set of  physical principles {\it without imposing} on the system dynamics. For the dynamics of center-of-mass, as long as system parameters are strictly positive, all initial center-of-masses tend to the same point, moreover, rotation matrices will tend to the same rotation matrix exponentially fast, as long as the initial data is sufficiently small. Moreover, due to the gradient flow structure of the Lohe matrix model, all initial configurations tend to a phase-locked state asymptotically. We believe that our work might be useful for a dynamical systems approach for the shape matching arising from computer science community e.g., \cite{A-C-K-Y, A-C, A-S-S, C-L, V-H}. We will leave several interesting bridges between aggregation modeling and shape matching as a possible future work.

\end{document}